%% file: arxiv.tex
\documentclass{article}
\usepackage{a4wide}
\usepackage[utf8]{inputenc}
\usepackage[T1]{fontenc}    % use 8-bit T1 fonts
\usepackage{hyperref}       % hyperlinks
\usepackage{url}            % simple URL typesetting
\usepackage{amsfonts}       % blackboard math symbols
\usepackage{nicefrac}       % compact symbols for 1/2, etc.
\usepackage{microtype}      % microtypography
\usepackage{xcolor} % colors
\usepackage{amsmath,amsthm,amssymb,bbm,graphicx,derivative}
\usepackage{todonotes}

\usepackage{cleveref}
\crefname{figure}{Figure}{Figures}
\crefname{equation}{Equation}{Equations}
\crefname{lemma}{Lemma}{Lemmas}
\crefname{theorem}{Theorem}{Theorems}
\crefname{section}{Section}{Sections}

\newtheorem{theorem}{Theorem}
\newtheorem{corollary}[theorem]{Corollary}
\newtheorem{lemma}[theorem]{Lemma}
\newtheorem{proposition}[theorem]{Proposition}
\newtheorem{definition}[theorem]{Definition}

\newcommand{\wbot}{\underline{w}}
\newcommand{\wtop}{\overline{w}}
\newcommand{\NC}{\operatorname{nc}}
\newcommand{\CR}[1]{\tilde{#1}}
\newcommand{\RR}{\mathbb{R}}
\newcommand{\EE}{\mathbb{E}}
\newcommand{\ZZ}{\mathbb{Z}}
\newcommand{\rob}{\operatorname{rob}}
\newcommand{\cost}{\operatorname{cost}}
\newcommand{\cons}{\operatorname{cons}}
\newcommand{\diff}{\mathop{}\!\mathrm{d}}

\newenvironment{ack}{%
  \section*{Acknowledgements}
}{%
}

\begin{document}

\author{
  Mathis Degryse\thanks{Sorbonne University, CNRS, LIP6, Paris, France}
    \and
  Imrane Zaakour\thanks{Centrale Supélec, Gif-sur-Yvette, Paris and ILLS, Montreal, Canada}
    \and
  Christoph Dürr\footnotemark[1]
    \and
  Spyros Angelopoulos\thanks{CNRS and International Laboratory on Learning Systems (ILLS), Montreal, Canada}
}

\input{main}

\end{document}

%% file: main.tex
\title{The Pareto Frontier of Randomized Learning-Augmented Online Bidding}

\maketitle
% import \ref{preview}/dashy-todo:0.1.3: todo import \ref{preview}/rubber-article:0.5.2: \textbf{ import \ref{preview}/equate:0.3.2: equate

%  }import images/draw\_seq.typ : plot\_seq, continuous\_seq, uniform\_sep import images/uniform\_sep\_smooth.typ: plot\_alpha\_beta\_uniform\_sep

\begin{abstract}
Online bidding is a classical problem in online decision-making, with applications in resource allocation, hierarchical clustering, and the analysis of approximation algorithms. We study its randomized learning-augmented variant, where an online algorithm generates a sequence of random bids while leveraging predictions from an oracle. 
We provide analytical upper and lower bounds on the optimal consistency $C$ as a function of the robustness $R$, which match when $R \geq 2.885$, effectively closing the gap left by previous work. 
The key technical ingredient is the notion of a {\em bidding function}, a novel abstraction that provides a unified framework for the design and analysis of randomized bidding strategies. We complement our theoretical results with an experimental application of randomized bidding to the incremental median problem, demonstrating the applicability of our algorithm in practical clustering settings.

% Online bidding is a classical problem in online decision-making, with applications in resource allocation, hierarchical clustering, and the analysis of approximation algorithms. We study its randomized learning-augmented variant, where an online algorithm generates a sequence of random bids while leveraging predictions from an oracle. 

% Our main result is an algorithm that achieves a provably near-optimal tradeoff between consistency and robustness, effectively closing the gap in prior work.
% The key technical ingredient is the notion of a {\em bidding function}, a novel abstraction that provides a unified framework for the design and analysis of randomized bidding strategies. This abstraction enables a more systematic understanding of the problem and leads to a principled analysis of randomized sequencing problems such as bidding.
\end{abstract}

\section{ Introduction}
\label{sec:introduction}

In several sequential decision-making settings, the algorithm relies on increasing a quantity until it reaches a level that guarantees success, without knowing that threshold in advance. For instance, a scheduler may gradually increase the runtime allocated to a job until it finishes, or an algorithm may successively guess larger values of an unknown optimum until feasibility is reached. In such settings, all unsuccessful choices contribute to the total cost, hence the performance of the algorithm depends on how quickly it reaches the threshold without excessive overshooting.

{\em Online bidding}~\cite{chrobak2008incremental} provides a simple abstraction of this fundamental tradeoff. In this problem, a player submits an increasing sequence of {\em bids} until one of them reaches or exceeds an unknown threshold $u$, and pays the sum of all bids made up to that point. The objective is to design a bidding strategy whose total cost is as small as possible relative to the unknown threshold. Formally, in the deterministic setting, the player produces an increasing sequence\footnote{As in previous work on online bidding and related incremental allocation problems, we consider {\em bi-infinite} sequences. This convention simplifies exposition and avoids additional technical assumptions, such as requiring that $u$ be at least as large as the first bid.} $X =(x_i)_{i \in \ZZ}$, and given a threshold $u$ it incurs cost equal to
\[
\cost(X,u)= \sum_{j\leq i} x_j \text{ for $i$ such that } x_{i-1}<u \leq x_i.
\]
The {\em normalized} cost of $X$ is then defined as $\NC(X,u)=\cost(X,u)/u$. For simplicity, we will use the notation $\NC(u)$ when $X$ is implied from context. In the {\em randomized} setting, $X$ is an increasing sequence of random bids, and the cost of $X$ on a threshold $u$ is measured by its expected value. 

Under the standard framework of {\em competitive analysis}~\cite{borodin1998online}, the objective is to minimize the worst-case normalized ratio, i.e., when the threshold is chosen adversarially. In practice, however, the algorithm may have access to additional information in the form of a {\em prediction}, obtained from historical data or learned models. This has motivated the study of {\em learning-augmented} algorithms~\cite{DBLP:books/cu/20/MitzenmacherV20}, where a prediction oracle is incorporated into the design and analysis of the algorithm. The goal is to leverage such predictions to obtain significantly improved performance when they are accurate, while still retaining strong worst-case guarantees when they are not.

We consider a natural oracle which provides the algorithm with a prediction $\hat{u}$. In this context, the {\em consistency} and the {\em robustness} of a bidding strategy are defined as its worst-case normalized cost, assuming error-free or adversarial predictions, respectively. Formally 
\begin{equation}
\cons(X)= 
\NC(X,\hat{u}), \quad \text{and} \quad
\rob(X)=\sup_{u} \NC(X,u).
\label{eq:cons-rob.formal}    
\end{equation}

A natural goal is to design algorithms that achieve the optimal tradeoff between consistency and robustness; in the learning-augmented literature, such algorithms are often called \emph{Pareto-optimal}. While this question has been answered for deterministic online bidding~\cite{angelopoulos_et_al_LIPIcs.ITCS.2020.52}, it remains open in the randomized setting. Randomization is known to improve performance even in the classical setting without predictions: the optimal randomized competitive ratio is 
$e$~\cite{chrobak2008incremental}, compared to the optimal deterministic ratio of 
4~\cite{beck:yet.more}. This naturally raises the question of how much randomization can improve learning-augmented bidding. In this work, we quantify this improvement by characterizing the best-possible tradeoff achievable by randomized algorithms.

%It is worth noting that for some problems, including online bidding, deterministic Pareto-optimal algorithms suffer from {\em brittleness}, in that an imperceptible prediction error may result in marked performance degradation. As shown in~\cite{elenter2024overcoming,Benomar-Perchet-tradeoffs-25} randomization helps alleviate the effect of brittleness since, at an intuitive level, the algorithm does not overfit the point prediction. For this reason, randomized learning-augmented bidding algorithms are more suitable for deployment in practical settings than deterministic ones, which further motivates their study.

\subsection{Contributions}
\label{subsec:contributions}

We present and analyze an algorithm that is $R$-robust for all $R\geq e$, and provably attains the  optimal consistency\footnote{During the preparation of this manuscript, we became aware of concurrent independent work by C. Lee et al.~\cite{Yongho-Shin-2026} that obtains a similar result.} for all $R\geq 2/\ln 2\approx 2.885$.  For the remaining, and very narrow  regime $R\in [e,2.885]$, we derive analytical upper and lower bounds which significantly improve upon the state of the art of~\cite{angelopoulos-simon-bidding-2025,shin2025improved}. In particular, for $R=e+\varepsilon$ and $\varepsilon \to 0$, the consistency of our algorithm is $C=e-\Omega(\varepsilon^{1/4})$, improving upon the previously known guarantee of $C=e-\Omega(\varepsilon)$. Our lower bound, which is tight for $R \geq 2.885$, also yields improved impossibility guarantees throughout the regime $R \in [e,2.885]$. Figure~\ref{fig:pareto} illustrates the obtained results and their comparison to previous work.

We complement our theoretical analysis with a numerical evaluation that explicitly incorporates prediction {\em error}. This study shows that our algorithm achieves improved \emph{smoothness}, namely a more graceful degradation in performance as prediction quality deteriorates, compared to previous approaches. Finally, we present an experimental application to the classical \emph{incremental median} problem, demonstrating that randomized online bidding can yield significant practical improvements in clustering settings, even under highly inaccurate predictions.

From a technical perspective, our design and analysis rely on a new abstraction that we call a \emph{bidding function}. This framework shifts the focus from discrete, randomized sequences to the analysis of continuous functions, providing substantially greater flexibility in both algorithm design and lower-bound arguments. In particular, it enables a more systematic treatment of consistency--robustness tradeoffs than previous sequence-based approaches. A useful parallel can be drawn with online {\em conversion} problems, where {\em threshold functions} have played a central role in the development of improved algorithms, and in unifying the analysis of several online problems~\cite{cao2020optimal,
sun_pareto-optimal_2021}.

\subsection{Related Work}
\label{subsec:related.work}

\noindent
{\bf Learning-augmented algorithms} \ 
Algorithms augmented with a machine-learned prediction on some unknown aspect of the input have been studied in a large variety of online problems, such as rent-and-buy problems~\cite{DBLP:conf/icml/GollapudiP19}, 
scheduling~\cite{lattanzi2020online}, caching~\cite{lykouris2021competitive}, matching~\cite{antoniadis2020secretary}, packing~\cite{im2021online}, covering~\cite{bamas2020primal}, sorting~\cite{bai-coester-sorting-2023}, graph problems~\cite{azar2022online},  optimal stopping~\cite{dutting2024secretaries} and data structures~\cite{lin2022learning}. These are only some representative works, and since the seminal works~\cite{lykouris2021competitive,purohit2018improving} the field has experienced a remarkable growth; we refer to the online compendium~\cite{ALPS-website} for a listing of several recent works.

The analysis of learning-augmented algorithms is intricate.
Several works have focused on tradeoffs between the consistency (i.e., the performance of the algorithm assuming error-free prediction) and the robustness (i.e., the performance under adversarial prediction). Tradeoffs for deterministic algorithms have been studied for a variety of problems, e.g.~\cite{sun_pareto-optimal_2021,DBLP:conf/eenergy/Lee0HL24,DBLP:journals/iandc/Angelopoulos23,bamas2020primal,almanza2021online, benomar2025pareto}, often yielding {\em Pareto}-optimal results; in contrast, randomized tradeoffs have  been limited to 
problems such as ski rental~\cite{wei2020optimal}, search games~\cite{doi:10.1287/opre.2024.1498}, $k$-server~\cite{DBLP:conf/aistats/ChristiansonSW23} and online bidding~\cite{angelopoulos-simon-bidding-2025,shin2025improved}, as we discuss below. 

\begin{figure}[t]
\centering
\includegraphics[width=\textwidth]{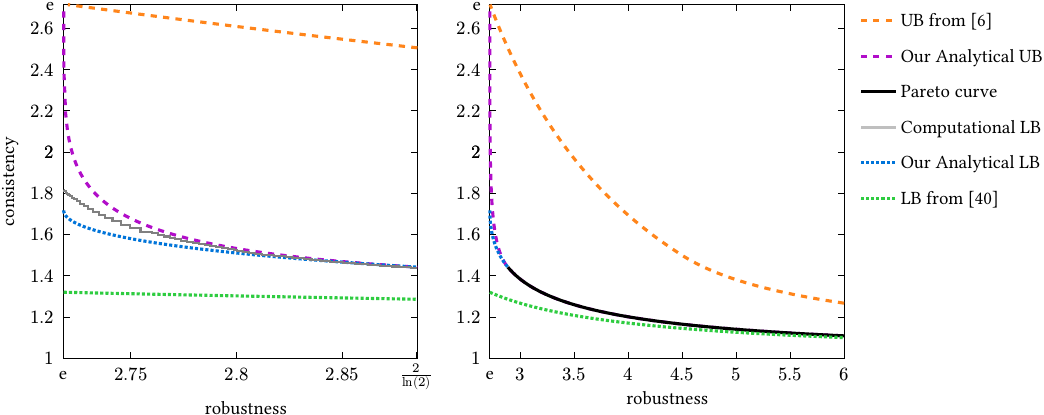}
 \caption{Consistency--robustness  tradeoff for various bidding strategies and different ranges of the robustness $R$. The dashed orange curve corresponds to the bidding sequence from \cite{angelopoulos-simon-bidding-2025}, while the dotted green curve corresponds to the lower bound given in \cite{shin2025improved}.
 UB and LB denote upper and lower bounds, respectively.}
 \label{fig:pareto}
\end{figure}

\noindent
{\bf Online bidding and related problems} \ 
For deterministic bidding, and in the absence of predictions,  a folklore result that goes back to linear search~\cite{beck:yet.more} shows that the optimal robustness (i.e., the optimal {\em competitive} ratio) is equal to 4, and is achieved by a simple doubling strategy of the form $X=(2^i)_{i \in \mathbb{Z}}$. Randomization improves the competitive ratio to $e$, which is tight as proven in~\cite{chrobak2008incremental} using a complex approach based on dual
fitting of a linear program. Online bidding has many applications in optimization problems such as incremental clustering and online $k$-median~\cite{charikar1997incremental}, load balancing~\cite{azar2005line}, searching for a hidden target in the infinite line~\cite{beck:yet.more}, and the design of algorithms with interruptible capabilities~\cite{RZ.1991.composing}; we refer to the survey~\cite{chrobak2006sigact} for a discussion of several applications.

Beyond worst-case competitive analysis, online bidding is one of the most studied problems under the learning-augmented lens. For deterministic algorithms,~\cite{angelopoulos_et_al_LIPIcs.ITCS.2020.52} 
gave Pareto-optimal algorithms, whereas~\cite{anand2021regression, im2023online} gave performance guarantees as a function of the prediction error. Concerning randomized strategies, a heuristic and a non-tight analytical lower bound were given in~\cite{angelopoulos-simon-bidding-2025}. In~\cite{Benomar-Perchet-tradeoffs-25}, randomization is leveraged implicitly by a perturbation of the prediction, though their algorithm is otherwise deterministic. The tradeoff between consistency and robustness was also studied in \cite{shin2025improved} under the name of the {\em button} problem, for which a non-tight analytic lower bound was presented.

\section{Preliminaries}
\label{sec:preliminaries}

We say that a bidding algorithm is $R$-robust and $C$-consistent if its robustness is at most $R$ and its consistency at most $C$. In our analysis, we will typically minimize the consistency, given a robustness requirement $R$. Note that $R$ must be at least the optimal randomized competitive ratio, namely at least $e$. 

Note that the normalized cost of a sequence is invariant under multiplication of the bids, a property that follows from the fact that the bidding sequences are bi-infinite. Formally, we have $\NC(\delta X, \delta u) = \NC(X, u)$ for any $\delta > 0$, where $\delta X:= (\delta x_i)_{i \in \ZZ}$. This implies that for any bidding sequence $X'$, we can find a scaled sequence $X$ such that
\begin{equation}
\cons(X)=\inf_u \NC(X,u) \
\text{ and }  \
\rob(X)=\sup_u \NC(X,u),
\label{eq:convenient_CR}
\end{equation}
which provides a convenient expression that we will be using in the remainder of this work. 

Given $R\geq e$, we define
\begin{equation}
\wbot(R) := -\frac{1}{W_{- 1}(- 1/R)} , \quad \wtop(R) := -\frac{1}{W_0(- 1/R)},
\label{eq:lamberts}
\end{equation}
where $W_0$ and $W_{-1}$ denote the principal and lower real branches of the Lambert function. We will use the shorthanded versions $\wtop,\wbot$ when $R$ is clear from context.

%We emphasize that for most of the sequences studied in this paper, the consistency is achieved at a particular value $u_0$, and by scaling it with the factor $\widetilde{u}/u_0$ we can guarantee that ratio at a predicted value $\widetilde{u}$. For sequences where the consistency $C$ is reached only in the limit, we know that for every $\varepsilon > 0$, we have $\NC(X, u_{\varepsilon})\leq C +\varepsilon$ for some $u_{\varepsilon}$, and we can guarantee this ratio at value $\widetilde{u}$, by scaling the sequence with the factor $\widetilde{u}/u_{\varepsilon}$.

The remainder of the paper is organized as follows. In Section~\ref{sec:functions}, we introduce the framework of bidding functions, establish its equivalence with bidding sequences, and illustrate the analysis of several classes of algorithms. Section~\ref{sec:upper} presents and analyzes our main algorithm, while in Section~\ref{sec:lower} we prove  the corresponding lower bounds using a linear programming approach.  Section~\ref{sec:numerical} concludes with a numerical comparison of algorithms that incorporates prediction error. 
Finally, in Section~\ref{sec:incremental}, we present an experimental application of randomized online bidding to the incremental median problem.

\section{ Bidding Functions : a Framework for Bidding Sequences}
\label{sec:functions}

In this section, we introduce the concept of a \textit{bidding function} that is instrumental in our analysis of the problem. Let $B : \RR \rightarrow \RR_+$ be a nondecreasing function such that $\lim_{t \rightarrow -\infty} B(t) = 0 , \lim_{t \rightarrow +\infty} B(t) = + \infty$ and for all $t \in \RR$, $\int_{-\infty}^t B(x) \diff x$ is finite.
 We call the domain of $B$ the set of \emph{reference points}.
We define the randomized bidding sequence {\em induced by} $B$ as $X_B\mathrel{:=}(B(i + \lambda))_{i \in \ZZ}$ where $\lambda$ is chosen uniformly at random in $\left[0 , 1\right]$.

In order to express the normalized cost of $X$ as a function of $B$, we introduce the concept of \emph{normalized mass} $\CR{B}$, defined by \[
\CR{B}(t) \mathrel{:=} \int_{-\infty}^{t + 1} \frac{B(x)}{B(t)} \diff x .
\] 
In addition, we define the \textit{work function} $w_B$, which will be useful in the analysis of our algorithms. This function  represents the expected sum over all bids up to the threshold $B(t)$, normalized by this threshold. \[
w_B (t) \mathrel{:=} \int_{-\infty}^t \frac{B(x)}{B(t)} \diff x .
\] 
When $B$ is clear from the context, we will sometimes use the shortcut $w_t := w_B(t)$.
To be able to relate $\CR{B}$ to $\NC(X_B)$, we define the \emph{generalized inverse} of $B$. Formally $B^{-}(u)$ is the infimum over all $t$ such that $B(t) \geq u$, see \cite{embrechts2013note}.

With the above definitions established, we can prove a strong connection between $\NC(X_B, u)$ and $\CR{B}$.

\begin{theorem}
\label{thm:bidding-fct}
For every $u \in \RR_{+}$ we have \begin{align*}
    \inf_u \NC(X_B, u) = \inf_t \CR{B}(t), \quad \text{and} \quad 
    \sup_u \NC(X_B, u) &= \sup_t \CR{B}(t).
\end{align*}
\end{theorem}

\cref{thm:bidding-fct} allows us to use bidding functions explicitly in the analysis. Hence, we can restate the consistency and robustness expressions~\eqref{eq:convenient_CR}  in terms of $\CR{B}(t)$ as
\begin{align}\label{eq:cons-rob-CRb}
\cons(B) \mathrel{:=} \inf_t \CR{B}(t), &\hspace{1cm}& \rob(B) \mathrel{:=} \sup_t \CR{B}(t).
\end{align}
% we decided to remove C^* from the paper

% Moreover we define the optimal consistency-robustness tradeoff as a function $C^{*}$ mapping $R$ to the best consistency over all $R$-robust bidding functions.

As an example, the exponential bidding function $B(t) = e^t$ induces the sequence $(e^{i + \lambda})_{i \in \ZZ}$ where $\lambda \sim U[0,1]$ which has optimal robustness  $e$ as shown in~\cite{chrobak2008incremental}. As a different example, the randomized algorithm of \cite{angelopoulos-simon-bidding-2025} is induced by the exponential of a piecewise linear function, as we will  discuss shortly. As a last example, a piecewise constant function such as $B(t) = 2^{\lfloor t\rfloor}$ induces the doubling sequence that has optimal deterministic robustness equal to 4~\cite{beck:yet.more}. These examples illustrate that bidding functions capture a broad spectrum of classical bidding strategies, and in fact this connection is complete: \emph{every} randomized bidding sequence is induced by some bidding function.

\begin{theorem}\label{thm:equivalence-seq-fct}
For every randomized sequence $X$, there is a bidding function $B$ such that the consistency (respectively,   robustness) of $X$ is upper bounded by the consistency (respectively,  robustness) of $B$. 
\end{theorem}

In Sections~\ref{sec:upper} and~\ref{sec:lower}, we leverage \cref{thm:equivalence-seq-fct} to obtain near-tight consistency--robustness tradeoffs. While these results require more intricate bidding functions, several interesting structural insights already arise from considering relatively simple ones. These examples help illustrate the main ideas and provide intuition for the more general techniques developed later. In particular, we consider the following classes of bidding functions:

\noindent
$\bullet$ \
Class $\cal E$ is comprised of all exponential bidding functions. This class contains, for instance, the $e$-robust bidding function $B(t)= e^t$ of ~\cite{chrobak2008incremental}. This warm-up class induces sequences of constant normalized cost, and reveals a connection between $R$-robustness and the expressions $\wtop,\wbot$, as defined in Section~\ref{sec:preliminaries}. 
%Class $\cal L$ belongs to the intersection of the next two classes $\cal D$ and $\cal I$. 

\noindent
$\bullet$ \ Class $\cal D$ contains all bidding functions $B$ that are exponential within each half-open interval of the form $[i , i + 1)$ for $i \in \ZZ$, and satisfy, in  addition, that $B(t + 1)/B(t)$ is constant. This class captures the bidding sequences studied in \cite{angelopoulos-simon-bidding-2025}. Our analysis allows us to obtain exact, closed-form expressions of the consistence--robustness tradeoffs, which was left open in~\cite{angelopoulos-simon-bidding-2025}.  

\noindent
$\bullet$ \ Class $\cal I$ contains all \textit{continuous} functions that are exponential within each half-open interval of the form $[i , i + 1)$ for $i \in \ZZ$. Our analysis proves, in particular, that if the robustness requirement $R$ is relatively small (i.e., smaller than 4.8), this class induces $R$-robust algorithms that have better consistency than the best known result of~\cite{angelopoulos-simon-bidding-2025}. We also prove that the tradeoff improvement becomes quantifiably more pronounced as $R \to e$.

We refer to the Appendix~\ref{app:classes} for the detailed discussion, with statements and proofs of results. We note that the class $\cal D$ consists of functions that are {\em non-continuous}, whereas the class $\cal I$ is comprised by linear {\em interpolations}. Together, these classes illustrate that achieving Pareto-optimality requires bidding functions with genuinely nontrivial structure.

\section{Upper Bounds} \label{sec:upper}

In this section, we present a bidding algorithm and analyze its performance. In Section~\ref{sec:lower}, we will prove that our algorithm is Pareto-optimal for almost all robustness values $R$. Our algorithm is defined in terms of a bidding function, given in Definition~\ref{def:bidding_function_main}.  Lemma~\ref{lemma:valid_main} will establish that this function is a valid bidding function and Theorems~\ref{thm:upper_main} and~\ref{thm:upper_second} will prove its theoretical guarantees. Last, Corollary~\ref{cor:A_asymptotic} will prove that the algorithm has much better consistency than state of the art algorithms even when $R$ is relatively small.

\begin{definition} \label{def:bidding_function_main}
    Given $R \geq 2 / \ln 2$, define $x := 1/R$, $\mu := R - \wtop$. Define the function $A$ \[
        A(t) = \begin{cases} 
          e^{(t-1) / \wtop} & \text{if } 1 \leq t \\
          1 & \text{if } 0 \leq t \leq 1, \\ 
          p_k(s) &  \text{if } t \leq 0 \text{ for } k=\lceil-t\rceil, s=[-t].
        \end{cases}
    \]
    % whereas for $t < 0$, $A$ is defined according to the rule \[
    %     A(s - k) = p_k (s) \text{ for } s \in [0,1) \text{ and } k \in \mathbb{N}.
    % \]
    
    Here the polynomial $p_k$ is defined inductively as follows:
    \begin{align*}
        p_0 (s) = 1, \quad p_1 (s) = x (\mu - 1 + s), \quad p_{k+1} (s) = p_k (0) - x \int_s^1 p_k (u) \diff u.
    \end{align*}
\end{definition}

We first give some intuition about this complex bidding function. The  function is chosen so that its normalized mass attains the minimum at $t = 0$, where the consistency is also evaluated, from~\eqref{eq:cons-rob-CRb}. For $t \geq 0$, the function induces first a deterministic bid (from $A(t) = 1$ for $t \in [0,1]$) followed by a randomized sequence that obeys an exponential rule as in class $\mathcal{E}$ (from $A(t) = e^{(t-1) / \wtop}$ for $t > 1$). For $t \leq 0$, $A(t)$ must be defined in a careful manner that is not symmetric to the case $t \geq 0$. First, there is a necessary discontinuity at $t = 0^-$. More importantly, for $t < 0$, $A(t)$ must obey a {\em delay} differential equation as demonstrated in Proposition~\ref{prop:delay}. 
The solution to this delay differential equation is expressed in terms of the sequence of polynomials given in Definition~\ref{def:bidding_function_main}. 

\begin{proposition} \label{prop:delay}
Let $R \geq e$ and $S = (-\infty,0) \backslash \ZZ_-$. Let $B$ be a bidding function continuous in $(-\infty,0)$ and differentiable in $S$ that satisfies the following delay differential equation in $S$ \[
    R B'(t) = B(t+1).
\]

Then, for $t < 0$, we have $\CR{B}(t) = R$.
\end{proposition}

The following lemma shows that $A$ satisfies the properties outlined in \cref{sec:functions}, hence it is a valid bidding function.

\begin{lemma} \label{lemma:valid_main}
$A$ is a valid bidding function.
\end{lemma}

The next theorem establishes our main upper bound. 

\begin{theorem} \label{thm:upper_main}
For $R \geq 2 / \ln 2$, it holds that $\cons(A) = R - \wtop$ and $\rob(A) = R$.
\end{theorem}

The case $R \leq 2 / \ln 2$ is technically more challenging. Leveraging the power of bidding functions, we solved a discretized LP so as to obtain intuition about the optimal bidding function. The solution reveals a similar structure to the case $R \geq 2 / \ln(2)$ with one significant difference: the function does not induce a deterministic bid anymore. This complicates the task of finding a tight upper bound. Nevertheless, we prove the following upper bound, whose proof is along the lines of the proof of \cref{thm:upper_main}, using a similar, but technically more involved definition of the bidding function $A(t)$.

\begin{theorem} \label{thm:upper_second}
Let $e \leq R < 2/\ln 2$, let $ \wtop$ be the largest solution of $w e^{1/w}=R$, and set $ \alpha:=1/\wtop$. Let $y\in(0,1]$ be the unique solution of $\alpha=y+\ln(2-y)$. Then \[
    \cons(A) = \frac{R}{2-y} \text{ and } \rob(A) = R.
\]
\end{theorem}

We also obtain the following corollary to \cref{thm:upper_second} that bounds the consistency of $A$ when $R$ is relatively close to the optimal optimized competitive ratio $e$. 

\begin{corollary} \label{cor:A_asymptotic}
    For the bidding function $A(t)$ used in the proof of \cref{thm:upper_second}, if $R = e + \varepsilon$ and $\varepsilon \to 0$, then \[
        \cons(A) \leq e - \Omega(\varepsilon^{1/4}).
    \]
\end{corollary}

We note that Corollary~\ref{cor:A_asymptotic} significantly improves upon the upper bound of \cite{angelopoulos-simon-bidding-2025} which showed that $C = R - \Omega(\varepsilon)$.

\section{Lower Bound} \label{sec:lower}

In this section, we derive a near-tight lower bound using a linear programming (LP) approach. We define a family of LPs parameterized by $R \geq e$ and integers $a, N, M \in \ZZ_+$. The variables, denoted by $x_k$ for $k \in [-N, M]$, represent the values of the bidding function at discrete reference points. The objective is to minimize the normalized mass at the reference point $t = 0$, which corresponds to the consistency. At the same time, the function is required to have normalized mass at most $R$ at every reference point that is an integer multiple of $1/a$ in the interval $[-N/a, M/a]$; these constraints encode the $R$-robustness requirement.

While LP-based approaches have previously been used to derive lower bounds for randomized online bidding~\cite{chrobak2008incremental,shin2025improved}, our formulation based on bidding functions is fundamentally different. Our LPs directly capture a discretized bidding function using a mesh of granularity $1/a$, leading to a simple and transparent representation of both consistency and robustness constraints. In contrast, the LPs of~\cite{chrobak2008incremental,shin2025improved} are based on discretized bidding sequences and use variables of the form $x_{i,j}$, representing the probability that bid $b_i$ precedes bid $b_j$ in the generated sequence. This sequence-based formulation leads to significantly more involved cost expressions, which limits the strength and tractability of the resulting bounds.

%The advantage of this formulation is that it describes directly the object that controls the cost ratios. The LP of Chrobak et al. is written at the level of the randomized bidding sequence itself: its variables $x_{i,k}$ encode the probability that two bids $b_i$ and $b_k$ are consecutive. This is a natural way to represent arbitrary randomized strategies, but the resulting variables do not directly correspond to the robustness and consistency quantities. They encode the internal structure of the sequence, and the cost constraints are then expressed indirectly through these transition probabilities. By contrast, once bidding functions have been introduced, the relevant information is simply the mass accumulated before a point, normalized by the value of the function at that point.

Formally, we define the linear program
$P_{a,N,M}^R$ with constraints labeled $\lambda,\gamma_k,\beta_k,\theta$:
\begin{align}
\min \quad &C \notag\\
\text{s.t.}\quad
&x_0\ge 1,\tag{$\lambda$}\\
&x_k\le x_{k+1}
& (-N\le k\le M-1) \tag{$\gamma_k$}\\
&\frac1a\sum_{j=-N}^{\min(M,k+a-1)}x_j\le R x_k
&(-N\le k\le M) \tag{$\beta_k$} \label{eq:beta-k}\\
&\frac1a\sum_{j=-N}^{a-1}x_j\le C \tag{$\theta$},\\
&x_k\ge0& (-N\le k\le M) \notag
\end{align}

The following lemma shows that $P_{a,N,M}^R$ is a lower-bound family of LPs.

\begin{lemma} \label{lemma:lower_bound_primal}
For every $ a,N,M \geq 1$ and $R \geq e$, the optimal objective value  of $P_{a,N,M}^R$ is a lower bound on the consistency of every $ R$-robust monotone bidding sequence. %In particular, if we denote this objective by $O_{a,N,M}^R$, we have:
%\[C^*(R) \ge O_{a,N,M}^R, \] where $C^*$ is the optimal consistency of every $R$ robust bidding function.
\end{lemma}

Let us denote by $D_{a,N,M}^R$ the dual of $P_{a,N,M}^R$, which is shown below. Note that the dual LP does not have a variable that corresponds to the primal constraint labeled  $(\theta)$, since it can be assumed to be $1$ by scaling. From weak duality, the objective of any feasible solution of the dual LP  is a lower bound on the optimal objective of $P_{a,N,M}^R$, hence also on the optimal consistency of $R$-robust bidding functions. 
\begin{align}
\max\quad &\lambda \notag\\
\text{s.t.} &\text{ for all} -N\leq k \leq M\notag\\
&
\frac{1}{a} \mathbf 1_{\{k\le a-1\}} + \frac1a\sum_{j = \max(k+1-a,-N)}^M \beta_j + \gamma_k-\gamma_{k-1} \ge \lambda \mathbf 1_{\{k=0\}}+R\beta_k, \tag{$x_k$} \\
& \lambda , \beta_k, \gamma_k \ge 0,\notag\\
& \gamma_M = \gamma_{-N-1} = 0 \notag.
\end{align}

\begin{proposition} \label{prop:weak_duality}
Let $R \geq e$ and $a,N,M \geq 1$. The objective value of every feasible solution to the dual LP  $D_{a,N,M}^R$ is a lower bound on the consistency of an $R$-robust bidding sequences.
\end{proposition}

We can now state the main result of this section.

\begin{theorem} \label{thm:lower_bound}
Let $R > e$, and let $ \wtop>1$ be the largest solution of $\wtop e^{1/\wtop}=R$. Then for every $\epsilon > 0$, every $R$-robust bidding sequence has consistency at least \[
    R-\wtop-\varepsilon.
\]
\end{theorem}

We conclude this section with a note that a computational solution to the lower-bound LP provides an improved lower bound in the robustness region $[e,2/\ln(2)]$; refer to  Fig~\ref{fig:pareto}.

%%%%%%%%%%%%%

%%%%%%%%%%%%%%%%

\section{Numerical Evaluation}
\label{sec:numerical}

% \todo{SA:we need to make the axis and the legends readable. Flip the figs.}
We present a numerical evaluation of the performance of our algorithms.
%\noindent{\bf Experimental setup.}
We consider the following setup. The oracle outputs a prediction $\hat{u}$, and the true threshold $u$ is generated at random according to $u=\hat{u}\cdot \exp(\eta)$, where $\eta \sim \mathcal{N}(0,\sigma ^2)$ represents a Gaussian noise. The variance $\sigma^2$ thus captures the prediction error of the oracle. Note that the median of $u$ is equal to $\hat{u}$. 

%\noindent{\bf Compared algorithms.}
We compare the performance of the following algorithms: the best algorithms from the class ${\cal D}$ and ${\cal I}$, defined in \cref{sec:functions}, as well as the algorithm $A$, which we analyzed theoretically in Theorems~\ref{thm:upper_main} and~\ref{thm:upper_second}. 
%In addition, we consider the $e$-competitive algorithm of~\cite{chrobak2008incremental} which has the same normalized cost at all reference points. 
We denote these three algorithms by {\sc D}, {\sc I}, and {\sc A}, respectively. Recall that algorithm {\sc D} is the state-of-the-art given in~\cite{angelopoulos-simon-bidding-2025}. 

%\noindent{\bf Performance plots.} 
\cref{fig:normalized_cost_fct_of_sigma2} depicts the expected normalized cost of the algorithms, as a function of the variance $\sigma^2$, for robustness $R=4$. Additional plots for other robustness values can be found in Appendix~\ref{app:numerical}. The crossed points depict the consistency attained by the algorithms.  

%\noindent{\bf Discussion.} 
%The plots show that both algorithms I and A achieve much better expected  normalized cost than algorithm D, as long as the prediction error is reasonably small. In particular, algorithm A has the smallest normalized cost as long as
% $\sigma^2\leq 1.57$, for $R=3$.$\sigma^2 \leq 1.19$, namely for a large range of prediction error.In addition, both algorithms A and I exhibit better {\em smoothness}, namely performance degradation as a function of the prediction error, than algorithm D, which approaches its worst-case normalized cost even for small values of the prediction error. 

The plots show that both algorithms {\sc I} and {\sc A} achieve significantly better expected normalized cost than algorithm {\sc D}, as long as the prediction error remains reasonably small. In particular, algorithm {\sc A} attains the smallest normalized cost whenever $\sigma^2 \leq 1.19$, covering a broad range of prediction errors. Moreover, both {\sc A} and {\sc I} exhibit substantially better \emph{smoothness}, that is, a more gradual performance degradation as the prediction error increases, compared to {\sc D}, whose normalized cost approaches its worst-case value even for relatively small prediction errors.

\begin{figure}[t!]
\centering
% \begin{tabular}{cc}
% \includegraphics[width=0.55\textwidth]{images/smoothness_R3.pdf} &
% \includegraphics[width=0.55\textwidth]{images/smoothness_R4.pdf} \\
% (a) R=4 & (b) R=3
% \end{tabular}
\includegraphics[width=.7\textwidth]{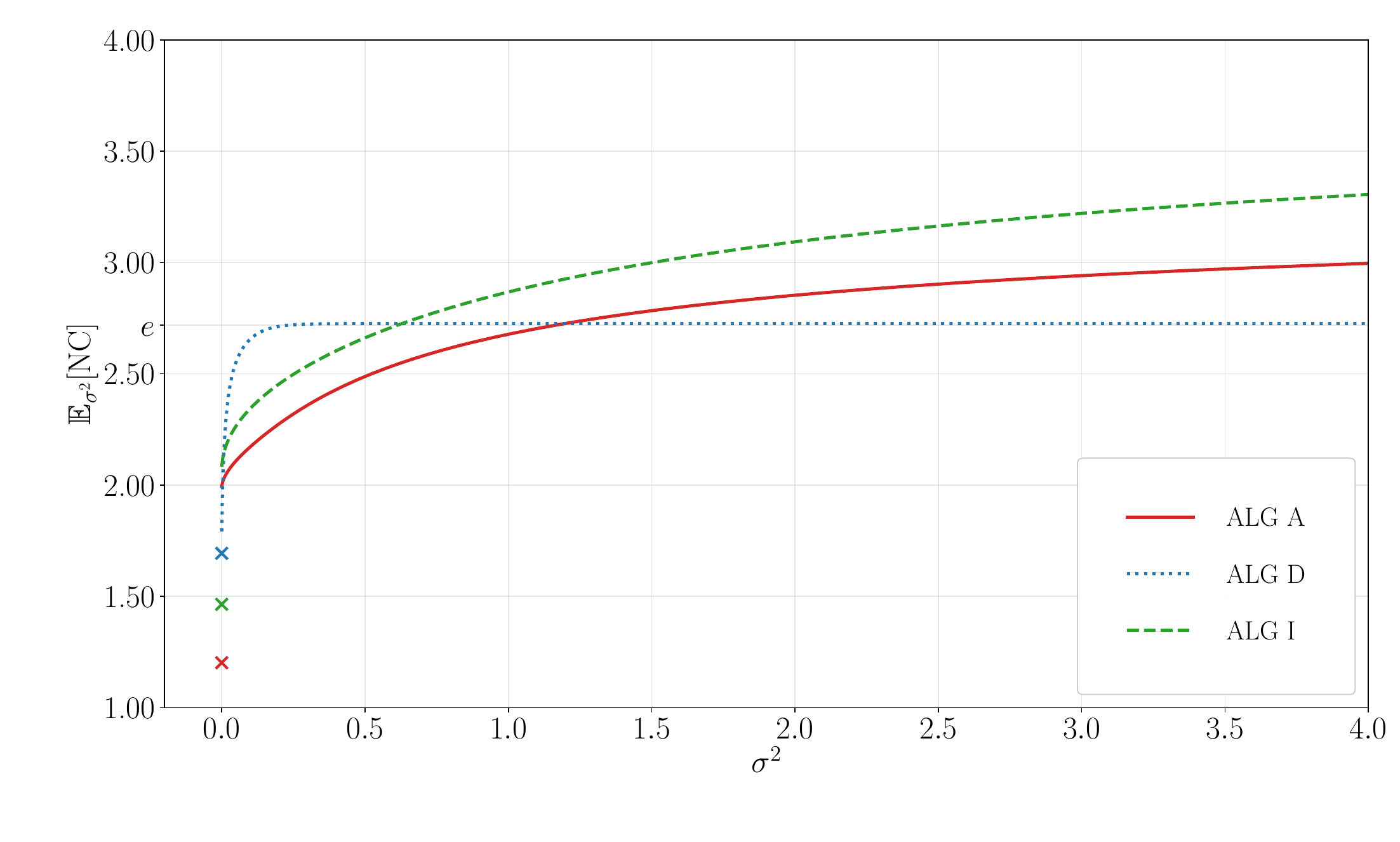}
\caption{Expected normalized cost as function of $\sigma^2 \in [0,4]$.}
\label{fig:normalized_cost_fct_of_sigma2}
\end{figure}

 %For $R=3$ and $\sigma^2 = 1.57$, our algorithm achieves the best expected nc where  \( u^\star \in  \left[0.127\,\hat u,\,7.854 \hat u\right]\) w.p. $0.9$. 

\input{rewriting}

\section{Conclusion}
\label{sec:conclusions}

We presented algorithms that provably achieve a near-optimal consistency--robustness tradeoff for randomized learning-augmented bidding, effectively closing the gap in previous work. A very small gap for $R \leq 2/\ln 2$ remains, which we believe can be further narrowed by allowing more resources to the 
computational lower bound. Central to our analysis is a novel function-based framework that maps bidding sequences to bidding functions and enables both the positive and impossibility results.
Beyond online bidding, this methodology extends naturally to closely related problems such as searching for a hidden target on the infinite line~\cite{beck:yet.more}, as shown in~\cite{angelopoulos-simon-bidding-2025}. We also aim to extend our approach to more challenging sequencing problems such as {\em contract scheduling}~\cite{RZ.1991.composing}, where the optimal randomized competitive ratio remains unknown even in the classical setting without predictions, making it a natural candidate for function-based analysis techniques.

% xtof: il ne faut pas décommenter la section ack.
% elle est cachée toute seule quand on compile pas avec l'option preprint
\begin{ack}
We used several generative AIs for this work, namely \emph{Emmy (Mistral.ai)} and \emph{Gemini 3 Flash} to assist us with programming and helping with the syntax of typst/cetz, as well as \emph{ChatGPT Plus} to help us with mathematical reasoning. However the answers produced by these AIs were used as starting points of truly personal work.

This work was partially funded by the project PREDICTIONS, grant ANR-23-CE48-0010 from the French National Research Agency (ANR).
\end{ack}

\bibliographystyle{plain} %parfait, merci Mathis. 
\bibliography{refs}
% ###################################################################################################################################################33

\clearpage
\input{appendix}

%% file: rewriting.tex
\section{Application to the Incremental Medians Problem}
\label{sec:incremental}

In this section, we conduct an experimental study of the application of randomized bidding to the incremental median problem, a classical online clustering problem. Previous work on this problem has been purely theoretical; we provide the first experimental results, in the context of a learning-augmented setting.

\noindent
{\bf Problem statement} \ 
Incremental medians is closely related to the well-studied $k$-median problem. Here, we are given a metric space $V$ with distance function $d$. In its simplest variant, the goal is to choose a set $F \subseteq V$ of $k$ \emph{facilities}. Each point $v \in V$ is assigned to its closest facility, which incurs cost defined as 
\[
\cost(F) \mathrel{:=} \sum_{v \in V} \min_{f \in F} d(v , f).
\]
The quality of a solution is measured by the {\em approximation ratio}, defined as $\cost(F) / \cost(F_k)$, where $F_k \subseteq V$ denotes an optimal set of $k$ facilities minimizing the cost. For a randomized algorithm producing a random set $F$, we consider the expected value of $\cost(F)$ divided by the $\cost(F_k^*)$.

The \emph{incremental median problem} introduces an additional constraint: the algorithm does not know in advance how many facilities it will be allowed to choose, and must therefore construct the solution sequentially, without the possibility of replacing previously selected facilities. Formally, a solution to the incremental median problem is a total order on $V$. We denote by $F_k$ the first $k$ facilities in that order, and by $\cost(F_k) / \cost(F_k^*)$ the approximation ratio of the algorithm for a specific value of $k$.

\noindent
{\bf Algorithms} \ 
Incremental medians was introduced under the name \emph{online median problem} in~\cite{mettu-plaxton-2003}, and its connection to the online bidding problem was studied in~\cite{chrobak2008incremental}. It was shown therein that, given a $\beta$-approximation algorithm for the $k$-median problem and a $\gamma$-competitive algorithm for the online bidding problem, one can obtain a $2\beta\gamma$-approximation algorithm for the incremental median problem. In Appendix~\ref{app:medians} we extend this connection to the learning-augmented setting, with a prediction $\hat{k}$ on the allowable number of facilities to choose. We show that any $R$-robust and $C$-consistent bidding algorithm yields an algorithm for incremental median that is $2\beta C$-consistent, and $2\beta R$-robust. We also address technical details related to the implementation.

\noindent
{\bf Datasets} \ 
For our experimental setup, we use the road network of the Isle of Man from the \textit{OpenStreetMap} project~\cite{openstreetmap}. Restricting to its largest connected component yields a graph on a vertex set $V$ of size 5340, where vertices correspond to geographic locations on the island. In this setting, facilities may represent, for example, mailboxes or electric charging stations, while the cost reflects the quality of service provided across the island. Edges correspond to road segments, and their weights represent traversal times. For the purpose of this study, we make all edges bidirectional by assigning to each undirected edge the minimum arc weight among both orientations. This connected edge-weighted graph induces a metric on $V$, where distances are defined by shortest-path lengths.

\noindent
{\bf Experimental evaluation} \ 
We choose $\hat{k} = 2500$ as the predicted number of facilities, and use bidding sequences optimized for low consistency at this value. Appendix~\ref{app:medians} provides additional plots for other values of $\hat{k}$.
Figure~\ref{fig:medians} depicts the average expected approximation ratio of various algorithms, based on different randomized bidding algorithms, as a function of the actual number of facilities $k$. We set the robustness requirement to $R=4$. The figure depicts the empirical ratios based on three different bidding algorithms: The best bidding function from class $\cal D$ introduced in~\cite{angelopoulos-simon-bidding-2025}, the best function from class 
${\cal I}$ in Section~\ref{sec:functions}, and the 
Pareto-optimal function $A$ analyzed in Theorem~\ref{thm:upper_main}. Following the same notation as in Section~\ref{sec:numerical}, we denote by {\sc D}, {\sc I} and {\sc A} the three algorithms.

We draw two main conclusions from Figure~\ref{fig:medians}. First, the theoretical consistency guarantees of these bidding functions are reflected in the same ordering of the consistency for incremental median: {\sc A} achieves the best consistency, followed by {\sc I} and then {\sc D}. More importantly, algorithm {\sc A} exhibits a pronounced improvement on the approximation ratio even when the prediction is significantly inaccurate. This indicates that the benefits of {\sc A} extend well beyond improved consistency--robustness tradeoffs, yielding practical gains even under large prediction errors.

The overall shape of the plots can be interpreted in light of the  definitions and the structural properties of the bidding functions. For {\sc D}, the plot is saw-like, similar to a deterministic plot, because the algorithm uses a relatively small amount of randomness. In contrast, the bidding functions $A$ and $I$ are similar for $k \in [1,\hat{k}]$, due to the presence of the same slope $1/\wtop$ in their exponent and the same deterministic bid at $k=\hat{k}$. For $k \in [\hat{k}, 5000]$, {\sc A} is markedly better than {\sc I}, because in this interval the latter uses few bids, due to the slope $1/\wbot$ in its exponent. 

\begin{figure}[t!]
    \centering
    \includegraphics[width=\linewidth]{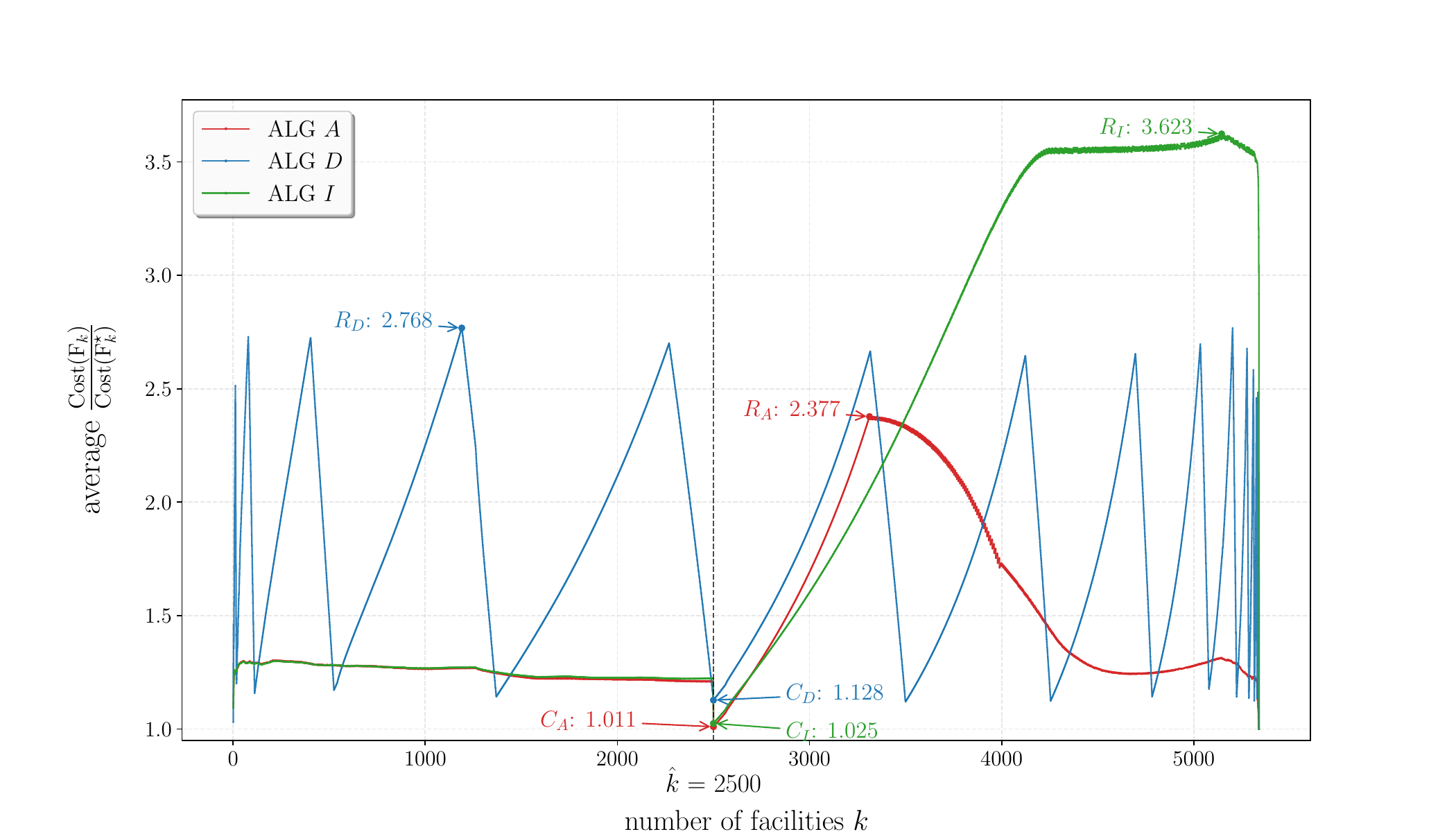}
    \caption{Average empirical approximation ratios for the incremental median problem applied to the Isle of Man road network, as a function of the actual number of facilities.
    $C_X$ and $R_X$ denote the empirical consistency and robustness of algorithm $X\in \{ {\sc A}, {\sc D}, {\sc I} \}$.
    }
    %We compare the Pareto optimal algorithm $A$ with the optimal bidding functions from classes $\cal D$ and $\cal I$ respectively. For an algorithm $X$, the graph shows $C_X$, the approximation ratio at the prediction, and $R_X$, the maximum approximation ratio.}
    \label{fig:medians}
\end{figure}

%% file: appendix.tex
\appendix 

{\Large \textbf{Appendix}}

\section{Omitted material from Section \ref{sec:functions}}

\begin{proof}[\textbf{Proof of Theorem \ref{thm:bidding-fct}}]
For $u \in \RR_{+}$ we have \begin{align}
\cost(X_B, u) & = \int_0^1 \sum_{i \in \ZZ} B(i + \lambda) \cdot \mathbbm{1}\left[B(i + \lambda - 1) < u\right] \diff \lambda \notag \\
 & = \sum_{i \in \ZZ} \int_0^1 B(i + \lambda) \cdot \mathbbm{1}\left[B(i + \lambda - 1) < u\right] \diff \lambda \notag \\
 & = \sum_{i \leq \lfloor B^-(u) \rfloor} \int_0^1 B(i + \lambda) \diff \lambda + \int_0^1 B(\left\lfloor B^-(u)\right\rfloor + 1 + \lambda) \mathbbm{1}\left[B(\left\lfloor B^-(u)\right\rfloor + \lambda) < u\right] \diff \lambda \notag \\
 & = \int_{-\infty}^{\left\lfloor B^-(u)\right\rfloor + 1} B(x) \diff x + \int_0^{\left[B^-(u)\right]} B(\left\lfloor B^-(u)\right\rfloor + 1 + \lambda) \diff \lambda \notag \\
 & = \int_{-\infty}^{\left\lfloor B^-(u)\right\rfloor + 1} B(x) \diff x + \int_{\left\lfloor B^-(u)\right\rfloor + 1}^{\left\lfloor B^-(u)\right\rfloor + 1 + \left[B^-(u)\right]} B(x) \diff x \notag \\
 & = \int_{-\infty}^{B^-(u) + 1} B(x) \diff x. \notag
\end{align} 

Hence we obtain that \begin{equation}
\NC(X_B, u) = \int_{-\infty}^{B^-(u) + 1} \frac{B(x)}{u} \diff x. \label{eq:cost_to_mass}
\end{equation}

Let $u \in \RR_{+}$, it holds that \begin{equation}
    \lim_{\substack{t \rightarrow B^-(u) \\ t > B^-(u)}} \CR{B}(t) \leq \NC(X_B, u) \leq \lim_{\substack{t \rightarrow B^-(u) \\ t < B^-(u)}} \CR{B}(t). \label{eq:bounded_nc_u}
\end{equation}

Indeed, both inequalities come from \eqref{eq:cost_to_mass} and the continuity of $\int_{-\infty}^{t + 1} B(x) \diff x$ in $t$. For the left inequality, we also use that for all $t > B^-(u)$ we have $B(t) \geq u$. Similarly, the right inequality holds because for all $t < B^-(u)$, $B(t) \leq u$. 

Let $t \in \RR$, we have \begin{equation}
 \NC(X_B, B(t)) \leq \CR{B}(t) \leq \lim_{\substack{u \rightarrow B(t) \\ u > B(t)}} \NC(X_B, u). \label{eq:bounded_cr_t}
\end{equation}

Similarly, this follows using \eqref{eq:cost_to_mass}. For the left inequality, we also use that if $u = B(t)$, then $B^-(u) \leq t$, and for the right inequality that $\lim_{\substack{u \rightarrow B(t) \\ u > B(t)}} B^-(u) \geq t$.

Combining the inequalities \eqref{eq:bounded_nc_u} and \eqref{eq:bounded_cr_t} completes the proof.
\end{proof}

\begin{proof}[\textbf{Proof of Theorem \ref{thm:equivalence-seq-fct}}]
Let $X$ be a randomized bidding sequence. For $v > 0$ let the random variable $N(v)$ be  the number of bids in $X$ between $1$ and $v$. If $v$ is smaller than $1$, it counts negatively the number of bids between $v$ and $1$. Formally  \[
N(v) \mathrel{:=} \sum_{k \in \ZZ} (\mathbbm{1}\left[X_k < v\right] - \mathbbm{1}\left[X_k < 1\right]).
\] 
% It could be possible to define $X$ such that $N$ is not defined, but in that case its expected normalized cost would not be defined, hence we can suppose that the definition of $N$ holds.

We denote the expected number of bids in $[1 , v)$, and its generalized inverse by $F(v)$ and $B(t)$ respectively: 
\begin{align*}
F(v) := & \: \EE [N(v)], & v > 0 , \\
B(t) := & \: \inf \{v > 0 | F(v) > t \}, & t \in \RR \\
% B^+(t) := & \: sup \{v < 0 | F(v) \leq t \}, & t \in \RR.
\end{align*} 
% $F$ must be defined for the same reason as $N$. 
% Otherwise we could not define the expected cost of $X$. 
% Even though $F$ might not be invertible, $B$ is well defined from the monotonicity of $F$. 
The cost of $X$ on threshold $u$ is
\[
\cost(X, u) = \sum_{X_k < u} X_k + Y_u,
\] where $Y_u$ is the first bid greater than or equal to $u$.

For all $u \geq 1$, we have: \begin{align}
\int_1^u F(x) \diff x & = \int_1^u \EE\left[N(x)\right] \diff x \notag \\
 & = \EE\left[\int_1^u N(x) \diff x\right] \notag \\
 & = \EE\left[\int_1^u \sum_{k \in \ZZ} (\mathbbm{1}\left[X_k < x\right] - \mathbbm{1}\left[X_k < 1\right]) \diff x\right] \notag \\
 & = \EE\left[\sum_{k \in \ZZ} \int_1^u (\mathbbm{1}\left[X_k < x\right] - \mathbbm{1}\left[X_k < 1\right]) \diff x\right] \notag \\
 & = \EE\left[\sum_{k \in \ZZ} \mathbbm{1}\left[X_k \in \left[1 , u\right)\right] \cdot (u - X_k)\right] \notag \\
 & = \EE\left[\sum_{u \leq X_k < 1} u - X_k\right] \notag \\
 & = u \EE\left[N(u)\right] - \EE\left[\sum_{1 \leq X_k < u} X_k\right]. \label{eq:F_integral_first}
\end{align}

Similarly, for all $u \leq 1$, we have: \begin{align}
\int_1^u F(x) \diff x & = - \int_u^1 \EE\left[N(x)\right] \diff x \notag\\
 & = u \EE\left[N(u)\right] + \EE\left[\sum_{u \leq X_k < 1} X_k\right]. \label{eq:F_integral_reversed}
\end{align}

We obtain from \eqref{eq:F_integral_first} and \eqref{eq:F_integral_reversed} that for all $u \geq v$: \[
\int_u^v F(x) \diff x = v \EE\left[N(v)\right] - u \EE\left[N(u)\right] - \EE\left[\sum_{u \leq X_k < v} X_k\right]. 
\]

By the area-inverse formula: \begin{align}
\int_{F(u)}^{F(v)} B(x) \diff x & = v F(v) - u F(u) - \int_{u}^v F(x) \diff x \notag\\
 & = \EE\left[\sum_{u \leq X_k < v} X_k\right]. \label{eq:B_bounded_integral}
\end{align}

Noticing that $\lim_{u \rightarrow 0} F(u) = -\infty$, it directly follows from \eqref{eq:B_bounded_integral} that \begin{equation}
    \int_{-\infty}^{F(v)} B(x) \diff x = \EE\left[\sum_{0 \le X_k < v} X_k\right]. \label{eq:B_integral}
\end{equation}

Next, we bound the expectation of $Y_u$. For convenience, we denote $s \mathrel{:=} F(u)$, and  introduce \[
G_u (y) := \mathbb{P}(Y_u < y).
\] Note that $\mathbbm{1}\left[Y_u < y\right] \leq N(y) - N(u)$. Taking expectations, we obtain \[
G_u (y) \leq F(y) - F(u) = F(y) - s.
\] 
We define the quantile of $Y_u$, i.e. for $0\leq t\leq 1$ let  \[
Q_u (t) \mathrel{:=} \inf \left\{y \geq u | G_u (y) > t\right\}, 
\] 
be the generalized inverse of $G_u$. If $G_u (y) > t$, then $F(y) > s + t$, hence $B(s + t) \leq y$. Taking the infimum on both sides, we obtain that $B(s + t) \leq Q_u (t)$. By integrating this inequality, \begin{equation}
    \int_s^{s + 1} B(x) \diff x \leq \int_0^1 Q_u (x) \diff x = \EE\left[Y_u\right]. \label{eq:B_step_integral}
\end{equation}

We can now complete the proof of the theorem. For an arbitrary $u \in \RR_+$, let $t = F(u)$, we have from \eqref{eq:B_integral} and \eqref{eq:B_step_integral} that:
\begin{align}
\CR{B}(t) 
&= \int_{-\infty}^{t+1} B(x) \diff x / B(t) \notag \\
&= \frac{1}{B(t)} (\int_{-\infty}^{F(u)} B(x) \diff x + \int_{F(u)}^{F(u) + 1} B(x) \diff x) \notag \\
&\leq \frac1u (\EE\left[\sum_{0 \le X_k < u} X_k\right] + \EE\left[Y_u\right]) \notag \\
&= \frac{\EE\left(\cost(X, u)\right)}{u} \label{eq:cr_t_smaller_nc_u}.
\end{align} 

For $t \in \RR$, let $u := B(t)$. For every $v > u$, we have $F(v) > t$, hence \[
\int_{-\infty}^{t+1} \frac{B(x)}{v} \diff x \leq \int_{-\infty}^{F(v)+1} \frac{B(x)}{v} \diff x \leq \frac{\EE \left[ \cost(X,v) \right]}{v}. 
\]

Letting $v$ tend to $u$, it follows that \begin{equation}
        \CR{B}(t) \leq \lim_{v \to u} \frac{\EE \left[ \cost(X,v) \right]}{v}. \label{eq:nc_u_bigger_cr_t}
\end{equation}

Combining both inequalities \eqref{eq:cr_t_smaller_nc_u} and \eqref{eq:nc_u_bigger_cr_t}, we obtain the wanted result, i.e. \[
    \inf_{t \in \RR} \CR{B}(t) \leq \inf_{u \in \RR_+} \NC(X_B, u), \quad \sup_{t \in \RR} \CR{B}(t) \leq \sup_{u \in \RR_+} \NC(X_B, u).
\]
\end{proof}

The following Lemma will be used in several proofs. 

\begin{lemma} \label{lemma:right_extension}
    Let $T \in \RR$ and $B$ a bidding function on $(-\infty, T)$. If the work of $B$ at $T$ is at most $w$, i.e. \[
        \int_{-\infty}^T B(x) \diff x \leq w B(T), 
    \]
    then one can extend the bidding function $B$ on $\RR$ by $B(t) = B(T) \exp((t - T) / w)$ so the normalized mass of $B$ is bounded on $(T, +\infty)$: \[
        \CR{B}(t) \leq w e^{\frac1w}.
    \]
\end{lemma}

\begin{proof}
    This extension is positive and increasing, hence the obtained $B$ is a valid bidding function on $\RR$. Let $t \geq T$, we have \begin{align*}
        \CR{B}(t) 
        &= \frac{\int_{-\infty}^{t+1} B(x) \diff x}{B(t)} \\
        &\leq  \frac{1}{B(T)} \exp(\frac{T-t}{w}) (\int_{-\infty}^{T} B(x) \diff x + \int_{T}^{t+1} B(x) \diff x) \\
        &\leq \exp(\frac{T-t}{w}) (w + w (\exp(\frac{t+1-T}{w}) - 1)) \\ 
        &= w \exp(1 / w).
    \end{align*}
\end{proof}

\subsection{Analysis of Bidding Function Classes $\cal E$, $\cal D$, $\cal I$} \label{app:classes}

%In order to study the robustness consistency tradeoff, it seems unavoidable to impose some structure on the considered randomized bidding sequences. Such restrictions will allow us to use some algebraic tools. This section is dedicated to different classes of bidding functions.  
 
\subsubsection{Class $\cal E$}
\label{sec:continuous}

An exponential bidding function $B$ is defined by $B(t) = e^{t / w}$, for some parameter $w > 0$ and for all $t \in \RR$. The normalized cost of a linear function is invariant with respect to the reference point. In particular, for any $t \in \RR$, we have \[
\CR{B}(t) = \CR{B}(0) = \frac{\int_{-\infty}^{1} e^{x/w} \diff x}{e^{0}} = w e^{1/w}.
\]

It is interesting to know for which values of $w$ the bidding function is $R$-robust. For this purpose, we observe that the ratio $w e^{1/w}$ is convex in $w$, which holds since its second derivative $e^{1/w} w^{- 3}$ is positive for $w > 0$. In addition its limit at $0$ from above is $\infty$, as is its limit at $\infty$. As a result, we have $w e^{1/w} \leq R$, for all $w \in \left[\wbot , \wtop\right]$ where $\wbot , \wtop$ are the two solutions of the equation $w e^{1/w} = R$. 
%See \cref{fig:wbot-wtop} for an illustration. 
Note that for the extreme case $R = e$, we have $\wbot = \wtop = 1$. 

%\begin{figure}[htbp] \centering \includegraphics[width=13cm]{images/xepx.pdf} \caption{Illustration of $\wbot , \wtop$ as a function of $R$. For $R = e$ we have $\wbot = \wtop = 1$ and for $R = R_0$ we have $\wtop = \wbot + 1$.} \label{fig:wbot-wtop} \end{figure}

\subsubsection{Class $\cal D$}

Bidding functions in class $\cal D$ are defined by two values $\ell , h \geq 0$ with $\ell + h > 0$, and have the form \[
B_{\ell , h} (t) = \exp \left( \lfloor t\rfloor \cdot (\ell + h) + \left[t\right] \cdot \ell \right) ,
\] 

where $\left[t\right] \mathrel{:=} t - \lfloor t\rfloor$ denotes the fractional part of $t$. When it is clear from context, we will omit the subscripts $\ell,h$ in the notation. We note that this class induces sequences equivalent to those studied in~\cite{angelopoulos-simon-bidding-2025}. In this section we give analytical results about the robustness--consistency tradeoff of these functions.

Note that the parameter $\ell$ allows to tune the amount of randomness in the sequence. In particular, of $\ell=0$, then the resulting sequence is deterministic.

\begin{lemma}
Let $\ell , h \geq 0$ such that $\ell + h > 0$. We have \begin{align*}
    \rob(B_{\ell,h}) &= \lim_{t \to 1^-} \CR{B}_{\ell,h}(t) = \frac{(e^{\ell} - 1)}{\ell} \frac{e^{\ell + h}}{(e^{\ell + h} - 1)} e^h \\
    \cons(B_{\ell,h}) &= \CR{B}_{\ell,h}(0) = \frac{(e^{\ell} - 1)}{\ell} \frac{e^{\ell + h}}{(e^{\ell + h} - 1)}.
\end{align*} \label{prop:linear-uniform-weakly}
\end{lemma}

\begin{proof}
%Let $\ell , h$ be the parameters from the proposition. 
Since bidding functions of class $\cal D$ satisfy that $B(t+1)/B(t)$ is constant, the normalized mass is $1$-periodic. For this reason, we can restrict the analysis to the interval $[0 , 1)$. We first compute $w_B (0)$.

By periodicity of the normalized mass, $w_B (0) = w_B (1)$, furthermore, if $\ell > 0$, then $w_B (1) = \frac{w_B (0) B(0) + \frac{e^{\ell} - 1}{\ell}}{B(1)}$. Thus we have $w_B (0) = \frac{e^{\ell} - 1}{\ell (e^{\ell + h} - 1)}$, which for convenience we denote $w$.

Next we compute \begin{align*}
\CR{B}(t) & = \frac{w B(1) + \int_1^{t + 1} B(x) \diff x}{B(t)} \\
 & = \frac{w e^{\ell + h} + e^{\ell + h}(\frac{e^{t \ell} - 1}{\ell})}{e^{t \ell}} \\
 & = e^{\ell + h} \left(\frac{1}{\ell} + e^{- t \ell} \left(w - \frac{1}{\ell}\right)\right).
\end{align*}

Since $\ell + h > \ell$, we have $w < 1 / \ell$ so this function is increasing in $[0 , 1)$. We finally obtain the equalities \[
    \rob(B) = \lim_{t \rightarrow 1^{-}} \CR{B}(t) = e^h (w + \frac{e^{\ell} - 1}{\ell}) = w e^{\ell + 2 h}, \quad \cons(B) = w e^{\ell + h}
\].

We note that the limits for $\ell \rightarrow 0$ of these expressions give $\rob(B) = \lim_{t \rightarrow 1^{-}} \CR{B}(t) = \frac{e^{2 h}}{e^h - 1}$ and $C = \frac{e^h}{e^h - 1}$, as in the deterministic case. 
\end{proof}

The obtained formulae for the robustness and consistency of bidding functions from class $\cal D$ make it possible to analyse the tradeoff they induce.

\begin{theorem}
Let $1 < C \leq e$. Let $R \mathrel{=} \min_{\ell , h} \left\{\rob(B_{\ell , h}) | \cons(B_{\ell , h}) \leq C\right\}$ be the optimal robustness among $C$-consistent bidding functions from class $D$. 

\begin{itemize}
 \item If $1 < C \leq \frac{3}{2}$, then $R = \frac{C^2}{C - 1}$, which is achieved by a deterministic bidding sequence. \vspace{1em}
 \item If $\frac{3}{2} \leq C \leq e$, then $R =\frac{\ell C^2 e^{-\ell}}{\ell C + 1 - e^{\ell}}$ for $\ell$ being the unique non negative root of \[
   1 - C \ell^2 - \ell - e^{\ell} (1 - 2 \ell).
 \]
\end{itemize}
\end{theorem}

\begin{proof}
From the proof of \cref{prop:linear-uniform-weakly} we have $R = w e^{\ell + 2 h}$ and $C = w e^{\ell + h}$ for $w =\frac{e^{\ell} - 1}{\ell (e^{\ell + h} - 1)}$. First we express $R$ in terms of $C$ and $\ell$ only. We have \begin{align*}
C & = \frac{e^{\ell} - 1}{\ell} \frac{e^{\ell + h}}{e^{\ell + h} - 1} & \Leftrightarrow \\
 C e^{\ell + h} & = \frac{e^{\ell} - 1}{\ell} e^{\ell + h} + C & \Leftrightarrow \\
 e^h \left(e^{\ell} C - \frac{e^{\ell} - 1}{\ell} e^{\ell}\right) & = C & \Leftrightarrow \\
 e^h (\ell C - (e^{\ell} - 1)) & = \ell C e^{-\ell} & \Leftrightarrow \\
 e^h & = \frac{\ell C e^{-\ell}}{\ell C + 1 - e^{\ell}} .
\end{align*} Hence $R = C e^h = \frac{\ell C^2 e^{-\ell}}{\ell C + 1 - e^{\ell}}$. For fixed $C$, we want to find the value of $\ell$ which minimizes $R$. We have \[
\partial_{\ell} R = \frac{C^2 e^{-\ell} \ell^2}{(C \ell - e^{\ell} + 1)^2} \frac{1 - C \ell^2 - \ell - e^{\ell} (1 - 2 \ell)}{\ell^2} .
\] For this product of two fractions, observe that the first fraction is always positive. Hence the sign of $\partial_{\ell} R$ equals the sign of the second fraction, which we denote by $F(C ,\ell)$. Using a series expansion, we obtain \[
F(C ,\ell) = \sum_{i = 0}^{+ \infty} \frac{3 + 2 i}{(i + 2)!}\ell^i - C .
\] We immediately obtain the following : 

\begin{itemize}
 \item $F$ is strictly increasing along $\ell$.
 \item $F(C , 0) \geq 0$ if and only if $C \leq \frac{3}{2}$.
 \item If $C > 3 / 2$, then $F$ is negative on $[0 ,\ell^{*})$ and positive on $(\ell^{*} , +\infty)$, where $\ell^{*}$ is the unique root of $F$.
\end{itemize}
Hence we have that, if $C > 3 / 2 $, then the optimal robustness is reached for the unique $\ell$ satisfying $F(C ,\ell) = 0$. Else, it is reached by the deterministic case $\ell = 0$.
\end{proof}

\subsubsection{Class $\cal I$} \label{sec:classI}

A bidding function $B$ in class $\cal I$ is defined by its value at integer reference points $(B(t))_{t \in \ZZ}$ for fractional reference points $t$, it is defined as the exponential interpolation between $B(\lfloor t \rfloor)$ and B($\lceil t \rceil)$.
In this section we use an equivalent representation based on a bi-infinite sequence of exponential \emph{slopes} $(\ell_i)_{i \in \ZZ}$, where $e^{\ell_i} = B(i + 1)/B(i)$. Without loss of generality we assume $B(0)= 1$, which allows us to express $B$ as \[
B(t) = \exp\left([t] \cdot\ell_{\lfloor t\rfloor} + \begin{cases}
 \displaystyle \sum_{i = 0}^{\lfloor t\rfloor - 1} \ell_i & \text{if } t \geq 0 \\
 \displaystyle -\sum_{i = \lfloor t\rfloor}^{- 1} \ell_i & \text{if } t < 0
\end{cases}\right).
\]

The following theorem is the consequence of the lower bound in \cref{lemma:pareto-class-I-lower} and the matching upper bound in \cref{lemma:pareto-class-I-upper}.
\begin{theorem} \label{thm:pareto-class-I}
Let $B$ be an $R$-robust bidding function from class $\cal I$ which minimizes its consistency. Let 
$(\ell_i)_{i \in \ZZ}$ be its associated sequence of slopes.

\begin{itemize}
 \item If $R \geq R_0$, then $\cons(B) = \wbot + 1$.
 \item If $R \leq R_0$, let $\ell^{*}$ such that $\wbot + e^{\ell^* - 1} = \wtop e^{\ell^*}$, then $\cons(B) = \wbot + \frac{e^{\ell^{*}} - 1}{\ell^{*}}$.
\end{itemize}
\label{theorem:lower_bound_continuous_linear}
\end{theorem}

For the analysis it will be important to bound the work $w_B(i)$ at every integer reference point $i \in \mathbb Z$.
% The sequence of slopes
% This function definition is not directly convenient to use. The goal of this section is to derive a characterization of $R$-robust linear continuous sequences. Since it is hard to obtain this result directly, we start by characterizing weakly $R$-robust linear continuous bidding functions.

% It is natural to narrow the study to integer reference points, since these functions are described by their sequence of slopes, each starting and ending on these points. More precisely, the normalized cost at these points is easily described by the sequence $(\ell_i)_{i \in \ZZ}$ and the induced sequence of works $(w_b (i))_{i \in \ZZ}$ which satisfies for every $i \in \ZZ$

% \[
% e^{\ell_i} \cdot w_b (i + 1) = w_b (i) + \frac{e^{\ell_i} - 1}{\ell_i} . 
% \]

\begin{lemma}
If a linear continuous bidding function $B$ is $R$-robust, then for every $i \in \ZZ , w_B (i) \in \left[\wbot , \wtop\right]$. Moreover, $\ell_i \leq 1/\wbot$ for every $i \in \ZZ$. \label{lemma:LC_bounded_work}
\end{lemma}

\begin{proof}
Let $(\ell_i)_{i \in \ZZ}$ be a sequence of slopes, $B$ the associated bidding function and $(w_B (i))_{i \in \ZZ}$ be the induced sequence of works. We define $w_i \mathrel{:=} w_B (i)$ for the ease of notation. Let us suppose that $B$ is $R$-robust.

For convenience, we define a function $f$, which computes the increase of the cost in one step, and $g$ normalizes $f$ to compute the change in the work:
\[
f(w , \ell) = w + \frac{e^{\ell} - 1}{\ell} , \hspace{1cm} g(w ,\ell) = f(w ,\ell)e^{-\ell} .
\]

Note that $f$ is increasing in its two arguments, and $g$ is increasing in $w$ but decreasing in $\ell$, both function are continuous on $[0 , +\infty)$, by extending $f$ with $f(w , 0) = w + 1.$

The normalized mass at integer points $\CR{B} (i)$ is equal to $w_i + \frac{e^{\ell_i} - 1}{\ell_i}$, so $R$-robustness translates for every $i \in \ZZ$ into 
\[
w_i + \frac{e^{\ell_i} - 1}{\ell_i} \leq R .
\]

We make the following claims: 

\begin{enumerate}
 \item If the sequence $(w_i)_{i \in \ZZ}$ converges (either when $i \rightarrow -\infty$, or $i \rightarrow + \infty$) towards $w^{*}$, then  $w^{*} \in \left[\wbot , \wtop\right]$
 \item If for some $k \in \ZZ$, $w_k > \wtop$ or $w_k < \wbot$, then $w_{k + 1} > w_k$.
 \item If for some $k \in \ZZ$, $w_k \geq \wbot$, then $w_{k + 1} \geq \wbot$.
\end{enumerate}

Next, we prove these claims: 

\begin{enumerate}
 \item By continuity of $g$, $w^{*}$ must satisfy that there exists $x$ such that $g(w^{*} , x) = w^{*}$, which gives $x = 1/w^{*}$. Thus the normalized cost at integer points must converge to $f(w^{*} , 1/w^{*}) = w^{*} e^{1/w^{*}} \leq R$. Which gives $w^{*} \in \left[\wbot , \wtop\right]$ by definition of these two values.
 \item Let $w_k$ be such a value. Note that $f(w_k , 1/w_k) > R$. We have $f(w_k , \ell_k) \leq R$, so $\ell_k  < 1/w_k$ by monotonicity of $f$. Hence $g(w_k , \ell_k) = w_{k + 1} > g(w_k , 1/w_k) = w_k$, by monotonicity of $g$.
 \item Let $w_k$ be such a value. Note that $f(\wbot , 1/\wbot) = R \geq f(w_k , \ell_k) \geq f(\wbot , \ell_k)$, hence $\ell_k \leq 1/\wbot$ by monotonicity of $f$. So $g(w_k , \ell_k) \geq g(\wbot , 1/\wbot) = \wbot$, by monotonicity of $g$.
\end{enumerate}
Next, we prove that for all $i \in \ZZ$, $w_i \leq \wtop$ : Suppose, by way of contradiction, that for some $k \in \ZZ$, $w_k > \wtop$. Then by (2), we have that for every $j \geq k$, $w_j < w_{j + 1}$, hence the sequence $(w_i)_{i \in \ZZ}$ is increasing. It is also upper bounded because $w_j \leq f(w_j , x) \leq R $, and therefore it must converges to some $w^{*} > w_k > \wtop$, which is not possible by (1).

We apply the same process to prove that for every $i \in \ZZ$, $w_i \geq \wbot$ : Let us suppose towards contradiction that some $w_k < \wbot$. By (2) and (3), it follows that the sequence $(w_i)_{i \leq k}$ is increasing, we also have it is lower bounded by $0$, so it must converge when $i$ goes to $- \infty$. Its limit $w^{*}$ must satisfy $w^{*} < \wbot$, which is not possible by (1). 
\end{proof}

The above Lemma formalizes an intuitive idea: if the work is too high at some point, then the robustness requirement cannot be met in the immediate future. Conversely, if the work is too low, it means the sequence is not dense in the past, hence the robustness requirement must have been broken at some point.

From now on, it will be important to study the normalized mass of a given bidding function in class $\cal I$, hence the next proposition.

\begin{proposition} \label{prop:class_I_normalized_mass}
    Let $B$ be a bidding function from class $\cal I$, let $(\ell_i)_{i \in \ZZ}$ be its sequence of slopes and $(w_i)_{i \in \ZZ}$ be its sequence of work at reference point $i \in \ZZ$. We have \[
    \CR{B}(t) = e^{(1 - \left[t\right]) \ell_{\lfloor t\rfloor}} \left(w_{\lceil t\rceil} + \frac{e^{\left[t\right] \ell_{\lceil t\rceil}} - 1}{\ell_{\lceil t\rceil}}\right).
    \]
\end{proposition}

\begin{proof}
For any $t \in \RR$, we have \begin{align*}
\CR{B} (t) & = \int_{- \infty}^{t + 1} B(x)/B(t) \diff x \\
 & = \int_{- \infty}^{\lceil t\rceil} B(x)/B(t) \diff x + \int_{\lceil t\rceil}^{t + 1} B(x)/B(t) \diff x \\
 & = \frac{1}{B(t)} \left(w_{\lceil t\rceil} B(\lceil t\rceil) + \int_{\lceil t\rceil}^{t + 1} B(\lceil t\rceil) e^{[x] \ell_{\lceil t\rceil}} \diff x\right) \\
 & = \frac{B(\lceil t\rceil)}{B(t)} \left(w_{\lceil t\rceil} + \int_0^{\left[t\right]} e^{x \ell_{\lceil t\rceil}} \diff x\right) \\
 & = e^{(1 - \left[t\right]) \ell_{\lfloor t\rfloor}} \left(w_{\lceil t\rceil} + \frac{e^{\left[t\right] \ell_{\lceil t\rceil}} - 1}{\ell_{\lceil t\rceil}}\right). \\
\end{align*}
\end{proof}

\begin{lemma}
Let $B$ be an $R$-robust bidding function from class $\cal I$, $(\ell_i)_{i \in \ZZ}$ its associated sequence of slopes, and $(w_B (i))_{i \in \ZZ}$ its sequence of works. Then $\cons(B)$ satisfies the following : 
\begin{itemize}
 \item If $R \geq R_0$, then $\cons(B) \geq \wbot + 1$.
 \item If $R \leq R_0$, let $\ell^{*}$ such that $g(\wbot , \ell^{*}) = \wtop$, then $\cons(B) \geq \wbot + \frac{e^{\ell^{*}} - 1}{\ell^{*}}$.
\end{itemize}
\label{lemma:pareto-class-I-lower}
\end{lemma}
\begin{proof}
We start by proving that we can modify the prefix  of a slope sequence to ensure its work is $\wbot$ on the modified prefix, without violating robustness, nor making the consistency worse.

Let $(\ell_i)_{i \in \ZZ}$ be a sequence of slopes that, without loss of generality, reaches its minimal normalized mass at $t \in \left[1 , 2\right]$. Let $(d_i)_{i \in \ZZ}$ be the same sequence of slopes except for $i \leq 0$ where $d_i = 1/\wbot$. We call $\mu$ and $\pi$  their associated bidding functions, respectively. First, note that the normalized mass of $\pi$ on $\left(-\infty , 0\right]$ is equal to $R$, and that for $t \in [1 , +\infty)$, we have $\CR{\mu}(t) \geq \CR{\pi} (t)$.

On $[0 , 1]$ we have using Proposition \ref{prop:class_I_normalized_mass} \begin{align*}
\CR{\pi} (t) & = e^{(1-t)/\wbot} \left(\wbot + \frac{e^{t \ell_1} - 1}{\ell_1}\right) \\
 & = \left(R - \frac{e^{1/\wbot}}{\ell_1}\right) e^{- t/\wbot} + e^{t (\ell_1 - 1/\wbot)} \frac{e^{1/\wbot}}{\ell_1}.
\end{align*}   

We recall that by Lemma \ref{lemma:LC_bounded_work}, $\wbot \leq 1/\ell_1$. Deriving this expression in $t$, we obtain 

\begin{align*}
(\CR{\pi})' (t) & = -\frac{1}{\wbot} e^{- t/\wbot} \left(\wbot e^{1/\wbot} + e^{1/\wbot} \frac{e^{t \ell_1} - 1}{\ell_1}\right) + e^{t(\ell_1 - 1/\wbot)} e^{1/\wbot} \\
 & = e^{t(\ell_1 - 1/\wbot)} \frac{e^{1/\wbot}}{\wbot} \left(\wbot - \frac{1}{\ell_1}\right) + \frac{e^{- t/\wbot}}{\wbot} e^{1/\wbot} \left(\frac{1}{\ell_1} - \wbot\right) \\
 & = \frac{e^{1/\wbot}}{\wbot} e^{- t/\wbot} \left(\wbot - \frac{1}{\ell_1}\right) (e^{t \ell_1} - 1) \leq 0.
\end{align*}

Since $\CR{\pi} (0) = R$, we have that $\CR{\pi} (t) \leq R$ for every $t \in \left[0 , 1\right]$, thus the bidding function $\pi$ is $R$-robust.

The consistency is reached for $t \in \left[1 , 2\right]$, using Proposition \ref{prop:class_I_normalized_mass} with $w_1 = g(\wbot, \ell_1) = e^{-\ell_1} (\wbot + \frac{e^{\ell_1} - 1}{\ell_1})$:

\[
\CR{\pi} (t) = \left((\wbot + \frac{e^{\ell_1} - 1}{\ell_1})e^{-\ell_1} + \frac{e^{(t - 1)\ell_2} - 1}{\ell_2}\right) \cdot e^{(2-t) \ell_1}.
\] 

Equivalently, for $t \in \left[0 , 1\right]$, \begin{align*}
\CR{\pi} (t + 1) & = \left(\wbot + \frac{e^{\ell_1} - 1}{\ell_1} + e^{\ell_1} \frac{e^{t \ell_2} - 1}{\ell_2}\right) \cdot e^{- t \ell_1} \\
 & = e^{\ell_1 (1 - t)} \left(w_2 + \frac{e^{t \ell_2} - 1}{\ell_2}\right).
\end{align*}

Deriving this expression in $t$, we obtain

\begin{align*}
(\CR{\pi})' (t + 1) & = \frac{e^{\ell_1 (1 - t)}}{\ell_2} \left(\ell_2 e^{\ell_2 t} - \ell_1(e^{\ell_2 t} + \ell_2 w_2 - 1)\right) \\
 (\CR{\pi})' (1) & = \frac{e^{\ell_1}}{\ell_2} (\ell_2 - \ell_1\ell_2 w_2) = e^{\ell_1} (1 - \ell_1 w_2) = 1 - \ell_1 \wbot \geq 0
\end{align*}

Hence the function $\CR{\pi}$ is increasing at $t = 1$. Noticing that the function $\CR{\pi}$ is bi-tonic on $[1,2]$ as being the sum of two exponential functions, we obtain that its minimum is reached on the boundary. Hence $\cons(\pi) = \min(\CR{\pi} (1), \CR{\pi} (2))$.

We have $\CR{\pi} (1) = f(\wbot , \ell_1)$ and $\CR{\pi} (2) = f(w_2 , \ell_2)$. Thus: 

\begin{itemize}
 \item $\cons(\pi) = \min(\CR{\pi} (1), \CR{\pi} (2)) \geq f(\wbot , 0) = \wbot + 1$, by monotonicity of $f$ in its two arguments. 
 \item If $\wtop - \wbot \leq 1$, then let $\ell_1^{*}$ be such that $g(\wbot , \ell_1^{*}) = \wtop$. We have $g(\wbot , \ell_1) = w_2 \leq \wtop = g(\wbot, \ell_1^*$. Therefore $\ell_1 \geq \ell_1^{*}$, so we have $\CR{\pi} (1) = f(\wbot , \ell_1) \geq f(\wbot , \ell_1^{*})$. Moreover, introducing $\ell^*_2$ with a similar definition, $g(w_2 , \ell_2^{*}) = \wtop = g(\wbot , \ell_1^{*})$, and since $w_2 \geq \wbot$, we must have $\ell_2^{*} \geq \ell_1^{*}$, so $\ell_2 \geq \ell_2^* \geq \ell_1^{*}$. We have $\CR{\pi} (2) = f(w_2 , \ell_2) \geq f(\wbot , \ell_1^{*})$.
\end{itemize}
\end{proof}

It remains to show the upper bound on the consistency of class $\cal I$ bidding functions. The following lemma will be useful, as it shows that a sequence of slopes can always be extended by the slopes $1/\wtop$.

\begin{lemma}           \label{lemma:pareto-class-I-upper}
Let $R \geq e$. 

\begin{itemize}
 \item If $R \geq R_0$, there exists a $R$-robust class $\cal I$ bidding function that achieves consistency $\wbot + 1$.
 \item Otherwise, there exists an $R$-robust class $\cal I$ bidding function that achieves consistency $\wbot + \frac{e^{\ell^{*}} - 1}{\ell^{*}}$, where $\ell^{*}$ is defined as in \cref{thm:pareto-class-I}.
\end{itemize}
\end{lemma}
\begin{proof}
For the case $R \geq R_0$, we consider the bidding function $b$ with slopes $(\ell_i)_{i \in \ZZ}$ defined by: \[
\ell_i = \begin{cases}
 1/\wbot &\text{if } i < 0 \\
 0 &\text{if } i = 0 \\
 1/\wtop &\text{if } i > 0.
\end{cases} 
\]

We want to show that $\CR{B}(t) \leq R$ for all $t \in \RR$. We consider the four following cases: 

For all $t \in (-\infty , - 1]$, we have $\CR{B} (t) = R$, due to the constant slope $1/\wbot$.

For all $t \in \left(- 1 , 0\right]$, we have $\CR{B} (t) = (\wbot + t + 1) e^{- t/\wbot}$, so $\CR{B} (- 1) = R$, $\CR{B} (0) = \wbot + 1$, and the derivative on this interval is $(\CR{B})'(t) = - \frac{t + 1}{\wbot} e^{- t/\wbot} \leq 0$, so the maximum normalized mass on this interval is $R$.

For all $t \in \left(0 , 1\right]$, $\CR{B} (t) = (\wbot + 1 + \wtop(e^{t/\wtop} - 1))$. This is clearly increasing, so the maximum on this interval is $\CR{B} (1) = \wbot + 1 + R - \wbot \leq R$, with equality when $R = R_0$.

For all $t \in (1 , +\infty)$, we use \cref{lemma:right_extension}. For this purpose we need to check $w_B(2) < \wtop$, which holds by $w_B(2) = g(w_B(1), \ell_1) = g(\wbot + 1 , 1/\wtop) \leq g(\wtop , 1/\wtop) = \wtop.$

As a result, this bidding function has robustness $R$ and consistency $\wbot + 1$.

For the case $R < R_0$, we consider the bidding function $B$ with sequence of slopes $(\ell_i)_{i \in \ZZ}$ defined by: 
\[
\ell_i = \begin{cases}
 1/\wbot &\text{if } i < 0 \\
 \ell^{*} &\text{if } i = 0 \\
 1/\wtop &\text{if } i > 0.
\end{cases} 
% \hspace{1cm} w_i = \begin{cases}
%  \wbot &\text{if } i \leq 0 \\
%  \wtop &\text{if } i \geq 1
% \end{cases} .
\]

We recall that by definition, $\ell^{*}$ satisfies $g(\wbot , \ell^{*}) = \wtop$. The normalized mass for $t \in (-\infty , - 1) \cup (1 , +\infty)$ can be treated exactly as in the case $R \geq R_0$.

We first show that $\ell^* \leq 1/\wtop$, this comes from monotonicity of $g$ and the inequality \begin{align*}
    g(\wtop, \ell^*) \geq g(\wbot, \ell^*) = \wtop = g(\wtop, 1/\wtop)
\end{align*}

If $t \in [- 1 , 0]$, then $\CR{B}(t) = e^{- t/\wbot}(\wbot + \frac{e^{(t + 1)\ell^{*}} - 1}{\ell^{*}})$, so the derivative gives \[
(\CR{B})' (t) = e^{- t/\wbot} (1 - \frac{1}{\wbot \ell^{*}}) (e^{(t + 1)\ell^{*}} - 1) \leq 0.
\] 

Hence the maximum normalized mass on this interval is reached at $-1$, for a value of $\CR{B}(-1) = R$.

If $t \in [0 , 1]$, then $\CR{B}(t) = e^{- t \ell^{*}} (\wtop + \wtop (e^{t/\wtop} - 1)) = \wtop e^{t (1/\wtop - \ell^{*})}$. We showed earlier that $\ell^* \leq 1/\wtop$, hence  $\CR{B}$ is increasing, therefore, the maximum normalized mass on the interval $[0,1]$ is $\CR{B}(1) = R e^{-\ell^*} \leq R$.

We showed that in both cases, the given bidding function $B$ of class $\cal I$ is $R$-robust and has the stated consistency.
\end{proof}

% \begin{figure}[htbp]
% \centering
%  ???
%  \caption{Competitive ratio of a continuous sequence with robustness $R = 4$ and consistency $C = 1.464$.}
% \end{figure}

\begin{theorem}
If $R = e + \varepsilon$ and $\varepsilon \to 0$, then the optimal consistency $C$ reached by $R$-robust class $\cal I$ bidding functions satisfies \[
C = e - O(\sqrt{\varepsilon}) + o(\varepsilon).
\]
\end{theorem}

\begin{proof}
We first recall that $\wbot(R),\wtop(R)$ values are $- 1/W_{- 1}(- 1/R)$ and $- 1/W_0(- 1/R)$.

Let $\varepsilon = R/e - 1$ so that $R = (1 + \varepsilon) e$. We have $- 1/R = -\frac1e (1 -\varepsilon + o(\varepsilon))$. Using the series expansion of the Lambert function from \cite{lambertw}, we have \begin{align*}
W_0     \left( -\frac1e (1 - u) \right) & \underset{u \to 0}{=} - 1 + \sqrt{2 u} + o(\sqrt{u}), \quad  
W_{- 1} \left( -\frac1e (1 - u) \right) & \underset{u \to 0}{=} - 1 - \sqrt{2 u} + o(\sqrt{u}).
\end{align*}

Hence, we obtain $\wbot = 1 - \sqrt{2 \varepsilon} + O(\varepsilon))$ and $\wtop = 1 + \sqrt{2 \varepsilon} + O(\varepsilon)$. We know that $\ell^{*}$ is solution of the equation \[
\wbot + \frac{e^{\ell} - 1}{\ell} = e^{\ell} \wtop .
\]

Let us define $\psi(\ell) = \frac{(1 - \ell)(e^{\ell} - 1)}{\ell} = \frac{e^{\ell} - 1}{\ell} - e^{\ell} + 1$, we have $\psi(\ell^{*}) = e^{\ell} (\wtop - 1) - (\wbot - 1)$. Since $\lim_{\varepsilon \rightarrow 0} \wtop = 1$ and $\lim_{\varepsilon \rightarrow 0} \wbot = 1$, we have that $\psi(\lim_{\varepsilon\rightarrow 0} \ell^{*}) = 0$, hence $\lim_{\varepsilon\rightarrow 0} \ell^{*} = 1$.

We have, on a neighborhood of $0$ for $\varepsilon$, that \begin{align*}
\psi(\ell^{*}) & = \psi(1 + (\ell^{*} - 1)) \\
 & = \psi(1) + \psi(1) (\ell^{*} - 1) + o(|\ell^{*} - 1 |) \\
 & = (1 - e)(\ell^{*} - 1) + o(|\ell^{*} - 1 |) \\
 \text{and} \quad  \psi(\ell^{*}) & = e^{\ell^{*}} (\wtop - 1) - (\wbot - 1) \\
 & = e(1 + o(1))(\sqrt{2 \varepsilon} + O(\varepsilon)) + (\sqrt{2 \varepsilon} + O(\varepsilon)) \\
 & = (e + 1)\sqrt{2 \varepsilon} + o(\sqrt{\varepsilon}).
\end{align*}

By equating both expressions, we have that $\ell^{*} = 1 - \frac{e + 1}{e - 1} \sqrt{2 \varepsilon} + o(\sqrt{\varepsilon})$.

Finally, \begin{align*}
C & = \wtop e^{\ell^{*}} \\
 & = e (1 + \sqrt{2 \varepsilon} + O(\varepsilon)) \left(1 - \frac{e + 1}{e - 1} \sqrt{2 \varepsilon} + o(\sqrt{\varepsilon})\right) \\
 & = e - \frac{2 e \sqrt{2 \varepsilon}}{e - 1} + o(\sqrt{\varepsilon})) \\
 & = e - \frac{2 \sqrt{2 e}}{e - 1} \sqrt{R - e} + o(\sqrt{R - e}).
\end{align*}
\end{proof}

\section{Omitted Material from Section \ref{sec:upper}}

\begin{proof}[\textbf{Proof of Proposition \ref{prop:delay}}]
Let $F(t) := \int_{-\infty}^t B(x) \diff x$. The function $F(t+1) - R B(t)$ is constant on $(-\infty,0)$. Indeed, this function is continuous and its derivative on $S$ is $B(t+1) - R B'(t) = 0$. Now for $t \to -\infty$, $F(t+1)$ and $B(t)$ both tend to $0$, hence the said constant is $0$ and we obtain $F(t+1) = R B(t)$, i.e. $\CR{B}(t) = R$.
\end{proof}

\begin{proof}[\textbf{Proof of Lemma \ref{lemma:valid_main} ($R \geq 2/\ln(2)$)}]
By definition of the polynomials, the function $A$ is continuous, except  at $t=0$, and is monotone increasing in $[1,+\infty)$. It remains to show that it is monotone increasing in $(-\infty, 0]$ as well as positive.

For this purpose we define $q_k := p_k (0)$ for $k \geq 1$. In order to obtain a recurrence relation for the sequence $(q_k)$. We define the following formal power series for $s \in [0,1)$: 
\[ 
    G(z,s) := \sum_{k \geq 0} p_k (s) z^k.
\]

By solving for $G$, we obtain: 
\begin{align*}
    \pdv{G}{s}(z,s) 
    &= \sum_{k \geq 0} p_k'(s) z^k \\
    &= \sum_{k \geq 1} x p_{k-1}(s) z^k \\
    &= x z \sum_{k \geq 0} p_k(s) z^k \\
    &= x z G(z,s).
\end{align*}

For convenience, we denote $\lambda := x z$; multiplying the previous equality by the power series $e^{- \lambda s}$, we obtain 
\[ 
    e^{- \lambda s} \pdv{G}{s}(z,s) - \lambda e^{- \lambda s} G(z,s) = 0,
\]
which can be rewritten as: 
\[
    \pdv{e^{- \lambda s} G}{s}(z,s) = 0. 
\]

Integrating this equality from $0$ to $1$, it follows that
\begin{equation}
        G(z,1) = e^{\lambda} G(z, 0). \label{eq:Gz1_first}
\end{equation}

We also have 
\begin{align}
    G(z,1)  
    &= 1 + \mu x z + \sum_{k \geq 2} p_k(1) z^k \notag \\
    &= 1 + \mu x z + \sum_{k \geq 2} p_{k-1}(0) z^k \notag \\
    &= 1 + \mu x z + z (G(z,0) - 1). \label{eq:Gz1_second}
\end{align}

From \eqref{eq:Gz1_first} and \eqref{eq:Gz1_second} we retrieve the desired value, using $a := 1 - \mu x = \wtop / R$ \begin{equation*}
        G(z,0) = \frac{1 - a z}{e^{x z} - z}.
\end{equation*}

Proving that $q_k > 0$ for all $k \geq 0$ is equivalent to showing that the power series $H(z) := G(z e^{\alpha},0)$ has positive coefficients. Introduce $\alpha := 1/\wtop$ so that $R = e^\alpha / \alpha$. Thus we have \[
    H(z) = \frac{1 - z}{e^{\alpha z} - z e^{\alpha}}.
\]

Next, we compute the coefficients of the power series $1/H(z)$. The power series $H$ can be inverted since its constant coefficient is $1$.
\[ 
    \frac{1}{H(z)} = \frac{e^{\alpha z} - z e^{\alpha}}{1 - z} =  1 - \sum_{n \geq 1} \tau_n z^n,
\]
with $\tau_n = e^{\alpha} - \sum_{j = 0}^n \frac{\alpha^j}{j!} > 0$. Hence 
\[ 
    H(z) = \frac{1}{1 - \sum_{n \geq 1} \tau_n z^n} = \sum_{m \geq 0} \left(\sum_{n \geq 1} \tau_n z^n \right)^m = \sum_{m \geq 0} e^{\alpha m} q_m z^m 
\]
has only positive coefficients, which gives us that $q_k = p_k(0)$ is positive for all $k \geq 0$. From here, a quick induction shows that $A$ is positive and increasing on $(-\infty, 0)$, using the fact that $A(t+1) = R A'(t)$ on this interval. It only remains to prove that $A$ has a finite integral from $-\infty$ to $0$. Let us write $F(t) := \int_{-\infty}^t B(u) \diff u$. Noticing that $\int_0^1 p_k(u) \diff u = R (q_k - q_{k+1})$ (consequence from $p_k' = p_{k-1}$) we obtain: 
\[ 
    F(0) = \sum_{k \geq 1} R (q_k - q_{k+1}) \leq R q_1.
\] 
Hence, $A$ is positive, increasing and has a finite integral. Thus it is a valid bidding function. 
\end{proof}

\begin{proof}[\textbf{Proof of Theorem \ref{thm:upper_main}}]
Recall the notation $x := 1/R$ and $\mu := R - \wtop$. We will show that $\CR{B}(t)$ is at most $R$.

Recall that $A$ is solution of the delay differential equation $A(t+1) - R A'(t) = 0$ on $(-\infty,0)$, is continuous on the same interval, and differentiable on this interval except at integer points. Hence, by \cref{prop:delay}, we have 
\[
    \forall t \in (-\infty, 0), \CR{A}(t) = \frac{F(t+1)}{A(t)} = R.
\]

For the consistency, we want it to be reached at $t = 0$. Taking $t \rightarrow 0^-$, we obtain $F(1) = R A(0^-) = \mu$, hence, using $A(0) = 1$: \[
    \CR{A}(0) = F(1) = \mu = R - \wtop.
\]
    
On the interval $[0,1]$, the worst normalized ratio is reached for $t = 1$. since $F$ is increasing, \begin{align*}
    \forall t \in [0,1], \CR{A}(t) 
    &\leq \CR{A}(1) \\
    &= F(2) / A(1) \\ 
    &= F(1) + \wtop (e^{1/\wtop} - 1) \\
    &= R - \wtop + R - \wtop \\
    &= 2 (R - \wtop) \\
    &\leq R.
\end{align*}

The last inequality holds by the assumption $R \geq 2 / \ln(2)$. For $t \geq 1$, it holds that $\CR{A}(t) \leq R$ directly from \cref{lemma:right_extension}.
\end{proof}

The following is the definition of the bidding function $A$ for $R < 2/\ln(2)$.

\begin{definition} \label{def:A_small_R}
    Given $R < 2 / \ln 2$, define $\alpha:=1/\wtop$ and consider $y \in (0,1]$ solution to the equation $\alpha = y + \ln(2-y)$. Let $L:=1 - y/\alpha$, and $\mu:=e^y/\alpha$. Define the function $A$ \[
        A(t) = \begin{cases} 
          e^{(t-L) / \wtop} & \text{if } L \leq t \\
          1 & \text{if } 0 \leq t \leq L, \\ 
          p_k(s) &  \text{if } t \leq 0 \text{ for } k=\lceil-t\rceil, s=[-t].
        \end{cases}
    \]
    
    Here the polynomial $p_k$ is defined inductively as follows:
    \begin{align*}
        p_0(s) &= \begin{cases}
        1 & \text{if } 0\leq s\leq L,\\
        e^{\alpha(s-L)} & \text{if } L\leq s\leq 1
        \end{cases}\\
        p_1(s) &= x \left( \mu - \int_s^1 p_0(u) \diff u \right)\\ 
        p_{k+1} (s) &= p_k (0) - x \int_s^1 p_k (u) \diff u.
    \end{align*}
\end{definition}

\begin{proof}[\textbf{Proof of Lemma \ref{lemma:valid_main} ($R < 2/\ln(2)$)}]

The case $R = e$ is immediate since it gives $A(t) = e^t$. We can thus assume $e < R \leq 2/\ln 2$. We set $x := 1/R$, $\alpha := 1/\wtop$, $L := 1 - y/\alpha$, and $\mu:=e^y/\alpha$. Since $ \wtop e^{1/\wtop}=R$, we have $ R=e^\alpha/\alpha$. Moreover, the equation defining $y$ gives $e^\alpha=e^y(2-y)$. Hence \[
    R = \frac{e^\alpha}{\alpha} = \frac{e^y (2-y)}{\alpha} = (2-y) \mu.
\]

We start with the positivity of the function, as in the proof of Theorem \cref{thm:upper_main}. It suffices to prove positivity at the left endpoints, hence we define $q_k := p_k(0)$ for $k \geq 1$, and we will show that $q_k > 0$ for every $k \geq 1$. we introduce the formal power series \[
    G(z,s):=\sum_{k\geq1}p_k(s)z^k.
\]

Note that the sum starts at $k = 1$ instead of $0$, because $p_0$ is not constant anymore as in the proof of Theorem \cref{thm:upper_main}. The computation is formal in $z$, and the derivative is only with respect to $s$. Using $p_1' = x p_0$ and $p_k'=xp_{k-1}$ for $k \geq 2$, it follows that \[
    \pdv{G}{s}(z,s) = x z (p_0(s) + G(z,s)).
\]

Let $\lambda:= x z$. Multiplying the above by $ e^{-\lambda s}$, we obtain \[
    \pdv{}{s} (e^{-\lambda s} G(z,s)) = \lambda e^{-\lambda s} p_0(s).
\]

Integrating from $0$ to $1$, this gives \[
    e^{-\lambda}G(z,1)-G(z,0) = \lambda \int_0^1 e^{-\lambda u} p_0(u) \diff u.
\]

We now express $G(z,1)$ in terms of $G(z,0)$. Since $p_1(1) = x \mu$ and  $p_k(1) = p_{k-1}(0)$ for $k \geq 2$, we have \[
    G(z,1) = x \mu z + z G(z,0) = \lambda \mu + z G(z,0).
\]
Substituting this in the above equation and solving for $G(z,0)$, we obtain \[
    G(z,0) = \lambda \frac{\mu e^{-\lambda} - \int_0^1 e^{-\lambda u} p_0(u) \diff u}{1 - z e^{-\lambda}}.
\]

By rescaling $z$, define $H(z):=G(e^\alpha z,0)$. Since $R = e^\alpha / \alpha$, the parameter $\lambda = x z$ becomes $\alpha z$ after this rescaling. We compute the integral in the numerator using the two pieces of $p_0$: \[
    \int_0^1e^{-\alpha z u} p_0(u) \diff u = \int_0^L e^{-\alpha z u} \diff u + \int_L^1 e^{-\alpha z u} e^{\alpha(u - L)} \diff u.
\]

The first integral is equal to $(1-e^{-\alpha z L}) / (\alpha z)$. For the second one, using $\alpha(1-L) = y$, we obtain \[
    \int_L^1e^{-\alpha z u} e^{\alpha(u-L)} \diff u = \frac{e^{y - \alpha z} - e^{-\alpha z L}}{\alpha(1-z)}.
\]
Using also $\mu = e^y / \alpha$, a direct simplification gives $H(z) = \frac{J(z)}{D(z)}$, where \[
    D(z) := e^{\alpha z} - z e^\alpha \qquad J(z):= \frac{e^{y z}-(1-z)e^{\alpha z}-z^2e^y}{1-z}.
\]
Since $H(z)=\sum_{k \geq 1} e^{\alpha k} q_k z^k$, it remains to prove that all coefficients of $H$ are positive. For the denominator $D(z)$ we have \[
    \frac{D(z)}{1-z} = 1 - \sum_{n \geq 1} \tau_n z^n,
\]
where $\tau_n = e^\alpha - \sum_{j=0}^n \frac{\alpha^j}{j!} = \sum_{j \geq n+1} \frac{\alpha^j}{j!} > 0$. Therefore, since the above series has non zero constant coefficient, we can reverse it and obtain that \[
    \frac{1-z}{D(z)} = \frac{1}{1 - \sum_{n \geq 1} \tau_n z^n} = \sum_{m \geq 0} \left(\sum_{n \geq 1} \tau_n z^n \right)^m 
\]
has non negative coefficients. Hence, if \[
    \frac1{D(z)} = \sum_{n \geq 0} i_n z^n,
\]
then the sequence $(i_n)$ is positive and non decreasing, because $1/D(z)$ is obtained from $(1-z)/D(z)$ by multiplying by $1/(1-z)$. Next we write the numerator $J(z)$ as \[
    J(z)=\theta z- \sum_{n \geq 2} \sigma_n z^n,
\]
where $\theta = 1 + y - \alpha = 1 - \ln(2-y) > 0$, and, for $n \geq 2$, \[
    \sigma_n = \sum_{m \geq n+1} \frac{y^m}{m!} + \frac{\alpha^n}{n!} > 0.
\]
We will use the identity $\sum_{n \geq 2} A_n = \theta$. Indeed: \[
    \sum_{n \geq 2} \sigma_n = \sum_{m \geq 3} (m-2) \frac{y^m}{m!} + \sum_{n \geq 2} \frac{\alpha^n}{n!}.
\]

In the above expression the first sum is equal to $(y-2)e^y + 2 + y$, and the second one is equal to $e^\alpha - 1 - \alpha$. Thus, using that $e^\alpha = e^y(2-y)$, \begin{equation}
    \sum_{n\geq2} \sigma_n = (y-2) e^y + 2 + y + e^\alpha - 1 - \alpha = 1 + y - \alpha = \theta.
    \label{eq:theta_sigma}
\end{equation}

We can now prove the positivity of the coefficients of $H$. Let $H(z) = \sum_{k \geq 1} h_k z^k$. Since $H=J/D$, we have \[
    h_k = \theta i_{k-1} - \sum_{n=2}^k \sigma_n i_{k-n} = \sum_{n=2}^k \sigma_n \left(i_{k-1} - i_{k-n} \right) + i_{k-1} \sum_{n>k} \sigma_n
\]

using \eqref{eq:theta_sigma}. The first term is non negative because $(i_n)$ is non decreasing. The second term is strictly positive because $i_{k-1} > 0$ and $\sigma_n > 0$ for every $n \geq 2$. Hence $ h_k > 0$ for every $k \geq 1$. Since $h_k = e^{\alpha k} q_k$, we have $q_k > 0$ for every $k \geq 1$. From here, positivity of $A$ on $(-\infty,0)$ follows by induction. Indeed, if $p_k$ is positive on $[0,1]$, then $p_{k+1}' = x p_k > 0$, so $p_{k+1}$ is increasing. Its minimum is reached at $0$ and is equal to $q_{k+1} > 0$. Hence, every $p_k$ is positive.

It remains to verify that $\int_{-\infty}^0 A(x) \diff x$ is finite. From the recurrence of Definition \ref{def:A_small_R}, evaluated at $s = 0$, we have \[
    q_{k+1} = q_k - x \int_0^1 p_k(u) \diff u,
\]

so \[
    \int_0^1 p_k(u) \diff u = R (q_k - q_{k+1}).
\]

Since all $q_k$ are positive, summing both parts of the above equation gives \[
    \int_{-\infty}^0 A(u) \diff u = \sum_{k \geq 1} \int_0^1 p_k(u) \diff u \leq R q_1 < +\infty.
\]

Thus $A$ is a valid bidding function. 
\end{proof}

\begin{proof}[\textbf{Proof of Theorem \ref{thm:upper_second}}]
As for \cref{thm:upper_main}, we can use \cref{prop:delay}, which gives\[
    \CR{A}(t) = R, \quad \forall t \in (-\infty,0).
\]

Taking $t \to 0^-$ in the equality $F(t+1) = R A(t)$, we have $F(1) = R A(0^-) = \mu$. Since $A(0) = 1$, the consistency is \[
    \CR{A}(0) = F(1) = \mu = \frac{R}{2-y}.
\]

It remains to verify that $\CR{A}(t) \leq R$ for $t \geq 0$. If $0 \leq t \leq L$, then $A(t) = 1$, and \[
    F(t+1) = F(1) + \int_1^{t+1} e^{\alpha(u-L)} \diff u = \mu + \frac{e^{\alpha(t+1-L)} - e^{\alpha(1-L)}}{\alpha}.
\]

Using $\mu = e^y / \alpha$ and $\alpha(1-L) = y$, this becomes $F(t+1) = \frac{e^{y+\alpha t}}{\alpha}$. Thus 
\[
    \CR{A}(t) = \frac{e^{y + \alpha t}}{\alpha}.
\]

Since $t \leq L$, we have $y + \alpha t \leq y + \alpha L= \alpha$, and therefore $\CR{A}(t) \leq e^\alpha / \alpha = R$. Finally, if $t \geq L$, then $A(t) = e^{\alpha(t-L)}$, and the same computation gives $F(t+1) = e^{\alpha(t+1-L)} / \alpha$. Hence \[
    \CR{A}(t) = \frac{e^{\alpha(t+1-L)}}{\alpha e^{\alpha(t-L)}} = \frac{e^\alpha}{\alpha} = R.
\]

We have proved that the robustness is at most $R$, while the consistency is exactly $R/(2-y)$. This proves the theorem.
\end{proof}

\begin{proof}[\textbf{Proof of Corollary \ref{cor:A_asymptotic}}]
We first recall that if $R$ is close enough to $e$, by \cref{thm:upper_second}, the consistency of the bidding function $A$ is $R/(2 - y)$ where $y$ is solution of the equation $1/\wtop = y + \ln(2-y)$. Set $\varepsilon := R/e - 1$ and $s = \wtop - 1$. Since $R \to e$, we have $\varepsilon \to 0$ and $s \to 0$. From $\wtop e^{1/\wtop} = R$, it follows that 
\[
    1 + \varepsilon = (1 + s) \exp \left( \frac{1}{1+s} - 1 \right).
\]

Expanding the right-hand side at $s=0$, we obtain 
\[
    (1 + s) \exp \left( \frac{1}{1+s} - 1 \right) = 1 + \frac{s^2}{2} + O(s^3).
\]

Therefore $\varepsilon=\frac{s^2}{2}+O(s^3)$ and hence $s = \sqrt{2 \varepsilon} + O(\varepsilon)$. We have that $y$ tends to $1$ by below, hence we can write $y = 1 - q$ with $q \to 0^+$. From the definition of $y$ we have 
\[
    1 - q + \ln(1 + q) = \frac{1}{1+s}.
\]

Expanding both sides gives \[
    1 - \frac{q^2}{2} + O(q^3) = 1 - s + O(s^2).
\]

Thus $q^2 / 2 = s + o(s)$, so $q=\sqrt{2s}+O(s)$. Using the expansion of $s$, this gives \[
    q = \sqrt{2\sqrt{2\varepsilon}} + O(\varepsilon^{1/2}) = 2^{3/4}\varepsilon^{1/4} +O(\varepsilon^{1/2}).
\]

Finally, since $2-y=1+q$, we have  
\[
    \frac{R}{2-y} = \frac{e(1+\varepsilon)}{1+q} = e(1+\varepsilon)\left(1-q+O(q^2)\right).
\]

Since $q=O(\varepsilon^{1/4})$, we have $q^2=O(\varepsilon^{1/2})$. Therefore 
\[
    \frac{R}{2-y} = e - e q + O(\varepsilon^{1/2}).
\]

Substituting the expansion of $q$, we obtain 
\[
    \frac{R}{2-y} = e - 2^{3/4} e \varepsilon^{1/4} + O(\varepsilon^{1/2}).
\]

Since $\varepsilon=(R-e)/e$, this is exactly 
\[
    \frac{R}{2-y} = e - 2^{3/4} e^{3/4}(R-e)^{1/4} + O\left((R-e)^{1/2}\right).
\]
\end{proof}

\section{Omitted Materials from Section \ref{sec:lower}}

\begin{proof}[\textbf{Proof of Lemma \ref{lemma:lower_bound_primal}}]
Let $a \geq 1$ and $R \geq e$. Let $B$ be a bidding function of robustness $R$ and consistency $C$. Since $B$ is $C$ consistent, we have that $\inf_{t \in \RR} \CR{B}(t) = C$. Then, for any $\varepsilon > 0$, we can shift the bidding function so that $\CR{B}(0) \leq C + \varepsilon$.

For all $k \in \ZZ$, define $x_k := a \int_{k / a}^{(k+1) / a} B(u) \diff u$. Note that we have $x_k \geq B(k/a) > 0$, and that we can suppose, without loss of generality, that $x_0 = 1$, up to dividing the bidding function $B$ by the constant factor $x_0$. From the robustness conditions, we obtain 
\begin{align*}
    \frac{1}{a} \sum_{j = -N}^{\min(M,k+a-1)} x_j
    &\leq \frac{1}{a} \sum_{j = -\infty}^{k+a-1} x_j \\
    &= \int_{-\infty}^{k/a + 1} B(u) \diff u \\
    &\leq R B(k/a) \\
    &\leq R x_k.
\end{align*}
The same calculations, combined with the fact that $x_0 = 1$, give us that 
\[
    \frac{1}{a} \sum_{j = -N}^{\min(M,a-1)} x_j \leq \frac{1}{a} \sum_{j = -\infty}^{a-1} x_j \leq C + \epsilon.
\]

Thus the objective value of $P_{a,N,M}^R$ is at most $C + \varepsilon$.  
\end{proof}

We first prove two intermediate lemmas that will be important in the proof of the Theorem.

\begin{lemma} \label{lemma:polynomial_root}
    Let $a \geq 1$ and $R \geq e$. The smallest root in modulus of $X^a - a R (X - 1)$ is its unique real root $\rho \in (1, 1 + 1/a)$. Moreover, $\rho$ is a root of multiplicity $1$.
\end{lemma}

\begin{proof}
    We first prove that $\rho$ exists. Let $f(x) = x^a - a R (x - 1)$, we have $f(1) = 1$ and 
    \[
        f(1 + \frac1a) = (1 + \frac1a)^a - R \leq e - R < 0.
    \]
    
    Hence $\rho$ exists; we now show uniqueness of $\rho$. Indeed, $f'(x) = a (x^{a-1} - R)$ is negative on $[1,1 + 1/a]$. Since $f'$ is non-zero on this interval, $\rho$ is of multiplicity $1$.
    
    It remains to verify that there does not exist any root of smaller modulus. Suppose, by way of contradiction, that there exists a root $\eta = v \rho$ with $|v| \leq 1$ and $v \neq 1$. We have 
    \[
        v^a \rho^a = a R (v \rho - 1). 
    \]
    
    Using that $\rho^a = a R (\rho - 1)$, we obtain the following 
    
\begin{align*}
        & v^a = \frac{v \rho - 1}{\rho - 1} \\
        &\Longleftrightarrow v^a - 1 = \frac{ \rho (v-1) }{\rho - 1} \\
        &\Longleftrightarrow \frac{v^a - 1}{v-1} = \frac{ \rho }{\rho - 1}.
\end{align*}
    
    The above can be rewritten as $\sum_{i = 0}^{a-1} v^i = \frac{\rho}{\rho - 1}$. Note that left hand side has modulus at most $a$ by the triangular inequality $(|v^i| \leq 1, \forall i)$ , and the right hand side has modulus at least $1 + a$ since $\rho \in [1, 1 + 1/a]$, a contradiction. 
\end{proof}

\begin{lemma} \label{lemma:bounded_linear_seq}
Let $a \geq 1$. The equation
\begin{align} 
    a R b_n^{(N)} = 1 + \sum_{m=1}^{\min(N,n+a-1)} b_m^{(N)} \quad (1\le n\le N). \label{eq:bounded_linear_seq}
\end{align}
has a family of positive solutions indexed by $N \geq 1$ such that for all $n \geq 1$, we have \[
    b_n^{(N)} \to_{N \to +\infty} (\rho - 1) \rho^{n-1},
\]
where $\rho$ is defined in Lemma \ref{lemma:polynomial_root}.
\end{lemma}

\begin{proof}
Let us fix $a \geq 1$.

We start by justifying the existence of a  positive solution $(b_n^{(N)})$. We first introduce an infinite sequence that is solution a to the unbounded version of \eqref{eq:bounded_linear_seq}. Let $\rho$ be as defined in Lemma \ref{lemma:polynomial_root}. The candidate infinite sequence is $(\sigma_n)_{n \geq 1}$ where \[
    \sigma_n := (\rho - 1) \rho^{n-1}.
\]

This sequence indeed satisfies, using the definition of $\rho$ \begin{align} 
    1 + \sum_{j=1}^{n+a-1} \sigma_j
    &= 1 + (\rho - 1) \sum_{j = 0}^{n+a-2} \rho^j \notag\\
    &= \rho^{n+a-1} \notag\\
    &= a R \sigma_n. \label{eq:infinite_linear_recurrent}
\end{align}

Consider the map $F_N:\mathbb R_+^N\to\mathbb R_+^N$ defined by, $\forall 1 \leq n \leq N$ \[
(F_N(b))_n := \frac{1}{a R} \left(1+\sum_{m=1}^{\min(N,n+a-1)} b_m \right).
\]

First note that this map is monotone, in the sense that if $(u_n)_{n \leq N} \leq (v_n)_{n \leq N}$, then $F_N((u_n)_{n \leq N}) \leq F_N((u_n)_{n \leq N})$, where $\leq$ corresponds to term by term comparison.

Next, recall that $\sigma$ is solution of \eqref{eq:infinite_linear_recurrent}, hence we obtain, denoting by $\sigma^{(N)}$ the finite sequence $(\sigma_n)_{n \leq N}$ that \[
    F_N(\sigma^{(N)}) \leq \sigma^{(N)}. 
\]

Moreover, $F_N(0) = (\frac{1}{aR},...,\frac{1}{aR}) > 0$, hence the sequence $0,F_N(0),F_N^{(2)}(0),\dots$ is monotone. It is also bounded by $\sigma^{(N)}$ from above. Indeed, we have $0 \leq \sigma^{(N)}$, and using monotonicity of $F_N$, we obtain $F_N^{(k)} (0) \leq F_N^{(k)} (\sigma^{(N)}) \leq \sigma^{(N)}$.

Thus, $0,F_N(0),F_N^{(2)}(0),\dots$ is monotone and bounded from above, so it converges to a positive fixed point that we denote $b^{(N)}$, and which is smaller than $\sigma^{(N)}$. It remains to compute the limit of $b^{(N)}_n$ as $N$ tends to $+\infty$.

Let $n \geq 1$. For all $N \geq n + a - 1$, we have \[
    a R b_n^{(N)} = 1 + \sum_{m=1}^{n+a-1} b_m^{(N)}.
\]

We define $b_n^*$ as the limit of $b_n^{(N)}$ as $N$ tends to $+\infty$. Note that for every $n$ and $N \geq n + a - 1$, $(F_N^{(k)} (0))_n \leq \sigma^{(N)}_n$. As $N$ tends to $+\infty$, we obtain that \begin{equation}
    b^*_n \leq \sigma_n. \label{eq:bn_limit_below_sigma}
\end{equation} 

We also have \[
    a R b_n^* = 1 + \sum_{m=1}^{n+a-1} b_m^*,
\]

which is exactly \eqref{eq:infinite_linear_recurrent}. The goal is now to prove that \eqref{eq:bn_limit_below_sigma} is tight. We define the difference of these two sequences, i.e. $d_n := \sigma_n - b^*_n \geq 0$, $(d_n)$ is solution of the homogeneous version of \eqref{eq:infinite_linear_recurrent}: \[
    \forall n \geq 1, \quad a R d_n = \sum_{m=1}^{n+a-1} d_m.
\]

Subtracting this equation for $n$ and $n+1$ gives that $(d_n)$ is solution of \[
    \forall n \geq 1, \quad a R (d_{n+1} - d_n) = d_{n+a}.
\]

This linear recurrent equation has the characteristic polynomial $P(X) = X^a - a R (X - 1)$. Since $\rho$ is a root of multiplicity $1$ of $P(X)$, we obtain that there exists a constant $K$ and polynomials functions $(K_i)_{i \in I}$ for each other root $(\rho_i)_{i \in I}$ of  \[
    d_n = K \rho^n + K_i(n) \rho_i^n.
\]

However, $d_n = \sigma_n - b_n^*$ is bounded by $\sigma^n$, so $d_n = O(\rho^n)$. $\rho$ being the root of smallest modulus of $P$, we must have that every $K_i$ is equal to the null polynomial. Otherwise, we would have $|d_n| \sim |\rho_i|^n$ for some $i \in I$. Thus $d_n = K \rho^n$ for some $K \geq 0$. Let us suppose towards a contradiction that $K > 0$. It follows that 
\[
    a R K \rho^n = \sum_{m=1}^{n+a-1} K \rho^m.
\]

Dividing by $K$, and using that $a R = \frac{\rho^a}{\rho - 1}$, we obtain on the left hand side $\frac{\rho^{a+n}}{\rho - 1}$ and on the right hand side $\frac{\rho^{a+n} - \rho}{\rho - 1}$, which are different, a contradiction. Thus $K = 0$ and $b^*_n = \sigma_n$ for all $n$, hence $b^* = \sigma$, which completes the proof.
\end{proof}

\begin{proof}[\textbf{Proof of Theorem \ref{thm:lower_bound}}]
We construct a family of feasible solutions of the dual, indexed by $a$ and $N$. Their objective values will converge to $R-\wtop$. 

Fix $a \geq 1$ and $N \geq 1$, and choose $M=a-1$. Note that the term $\mathbf 1_{\{k\le a-1\}}$ is equal to $1$ for every $k \in \{-N,\dots,a-1\}$.

We first consider the sequence $b_1^{(N)}, \dots, b_N^{(N)}$ from Lemma \ref{lemma:bounded_linear_seq}. We specify the values of the dual variables. Let 
\begin{align*}
\beta_{-n} &:= b_n^{(N)} & (1 \le n \le N), \\
\beta_k &:= 0 & (0 \le k \le a-1).
\end{align*}

For $\gamma$, let $ \gamma_i=0$ for $ -N\le i\le -1$. For $ 0\le i\le a-2$, define \begin{equation}
    \gamma_i := \frac1a \sum_{m=i+1}^{a-1} \left( 1 + \sum_{j=1}^{a-1-m} \beta_j^{(N)} \right), \label{eq:gamma_def}
\end{equation}

keeping the boundary convention $\gamma_{a-1}=0$. Finally, we specify \[
\lambda_{a,N} := 1+ \frac1a \sum_{m=1}^{a-1} (a-m) b_m^{(N)}. 
\]

We will prove that the above specify a feasible dual solution. 

First consider a negative index $-N \leq k \leq -1$. Since $\gamma_{k} = \gamma_{k-1} = 0$, the dual constraint that corresponds to $k$ becomes \[
\frac1a + \frac1a \sum_{j=\max(k+1-a,-N)}^{-1} \beta_j \ge R \beta_k^{(N)}.
\]

Substituting $j = -m$ and $n = -k$, this becomes \[
\frac1a \left( 1 + \sum_{m=1}^{\min(N,n+a-1)} b_m^{(N)} \right) \ge R b_{n}^{(N)}.
\]

By \eqref{eq:bounded_linear_seq} from Lemma \ref{lemma:bounded_linear_seq}, the above inequality holds with equality.

Next, consider $1 \le k \le a-1$. Then $\beta_k = 0$, and there is no $\lambda$-term. The left-hand side of the dual constraint for $k$ is \[
\frac1a + \frac1a \sum_{j=k+1-a}^{-1} \beta_j + \gamma_k - \gamma_{k-1}.
\]

In the above expression, the sum over $\beta$ is equal to $\sum_{m=1}^{a-1-m} b_m^{(N)}$. By Equation \eqref{eq:gamma_def}, we have \[ 
\gamma_{k-1}-\gamma_k = \frac1a \left( 1 + \sum_{m=1}^{a-1-m} b_n^{(N)} \right).
\]

Thus the left-hand side of the constraint is equal to $0$, so it holds with equality.

It remains to verify the case $k=0$. Since $\beta_0 = 0$ and $\gamma_{-1}=0$, the left-hand side of the constraint is \[
\frac1a + \frac1a \sum_{m=1}^{a-1} b_m^{(N)} + \gamma_0.
\]

Using Equation \eqref{eq:gamma_def}, this is \[
\frac1a + \frac1a \sum_{m=1}^{a-1} b_m^{(N)} + \frac1a \sum_{m=1}^{a-1} \left( 1 + \sum_{j=1}^{a-1-m}b_j^{(N)} \right).
\]

The term $b_m^{(N)}$ appears $a-m$ times in total in the above expression. Therefore the left-hand side of the constraint for $k=0$ becomes \[
1 + \frac1a \sum_{m=1}^{a-1} (a-m) b_m^{(N)} = \lambda_{a,N}.
\]

Hence the dual constraint for $k=0$ also holds with equality. We conclude that $(\beta,\gamma,\lambda)$ is a feasible solution to the dual $D_{a,N,M}^R$.

The best lower bound is obtained as the parameters $a$ and $N$ tend to $+\infty$. We now compute the limit of the objective $\lambda_{a,N}$. For fixed $a$, let $\rho_a$ be the root from Lemma \ref{lemma:polynomial_root}. We denote $\lambda_a = \lim_{N \to\infty} \lambda_{a,N}$ and obtain \begin{align} 
    \lambda_a 
    &= 1 + \frac1a \sum_{n=1}^{a-1} (a-n) \rho_a^{n} - \rho_a^{n-1} \label{eq:lambda_limit} \\ 
    &= 1 + \frac1a \sum_{n=1}^{a-1} \rho_a^n - (a - 1) \notag \\
    &= \frac1a \sum_{n = 0}^{a-1} \rho_a^n \notag,
\end{align} 

By defining $\theta_a$ so that $\rho_a = e^{\theta_a/a}$, the definition of $\rho_a$ gives \[
e^{\theta_a}= a R \left( e^{\theta_a/a} - 1 \right).
\]

As $ a\to\infty$, since $\rho_a$ tends to $1$, $\theta_a / a$ tends to $0$ and this gives $e^\theta=R\theta$ where $\theta = \lim_{a \to \infty} \theta_a$. Since $\rho_a \in (1,a/(a-1))$, the relevant limit satisfies $\theta\in(0,1]$, and the equation $ e^\theta/\theta=R$ gives $\theta=1/\wtop$.

Finally, \eqref{eq:lambda_limit} becomes a Riemann sum: 
\[
\lambda_a = \frac1a \sum_{m=0}^{a-1} e^{\theta_a m/a} \longrightarrow \int_0^1 e^{\theta x} \diff x = \frac{e^\theta-1}{\theta}.
\]

From the fact that $ e^\theta/\theta=R$ and $ \theta=1 / \wtop$, we obtain 
\[
\frac{e^\theta-1}{\theta} = \frac{e^\theta}{\theta}-\frac1\theta = R-\wtop.
\]

Therefore every $R$-robust bidding sequence has consistency at least $R - \wtop - \varepsilon$. This proves the theorem.
\end{proof}

\section{Additional numerical analysis}
\label{app:numerical}

\begin{figure}[t!]
    \centering
    \includegraphics[width=0.8\linewidth]{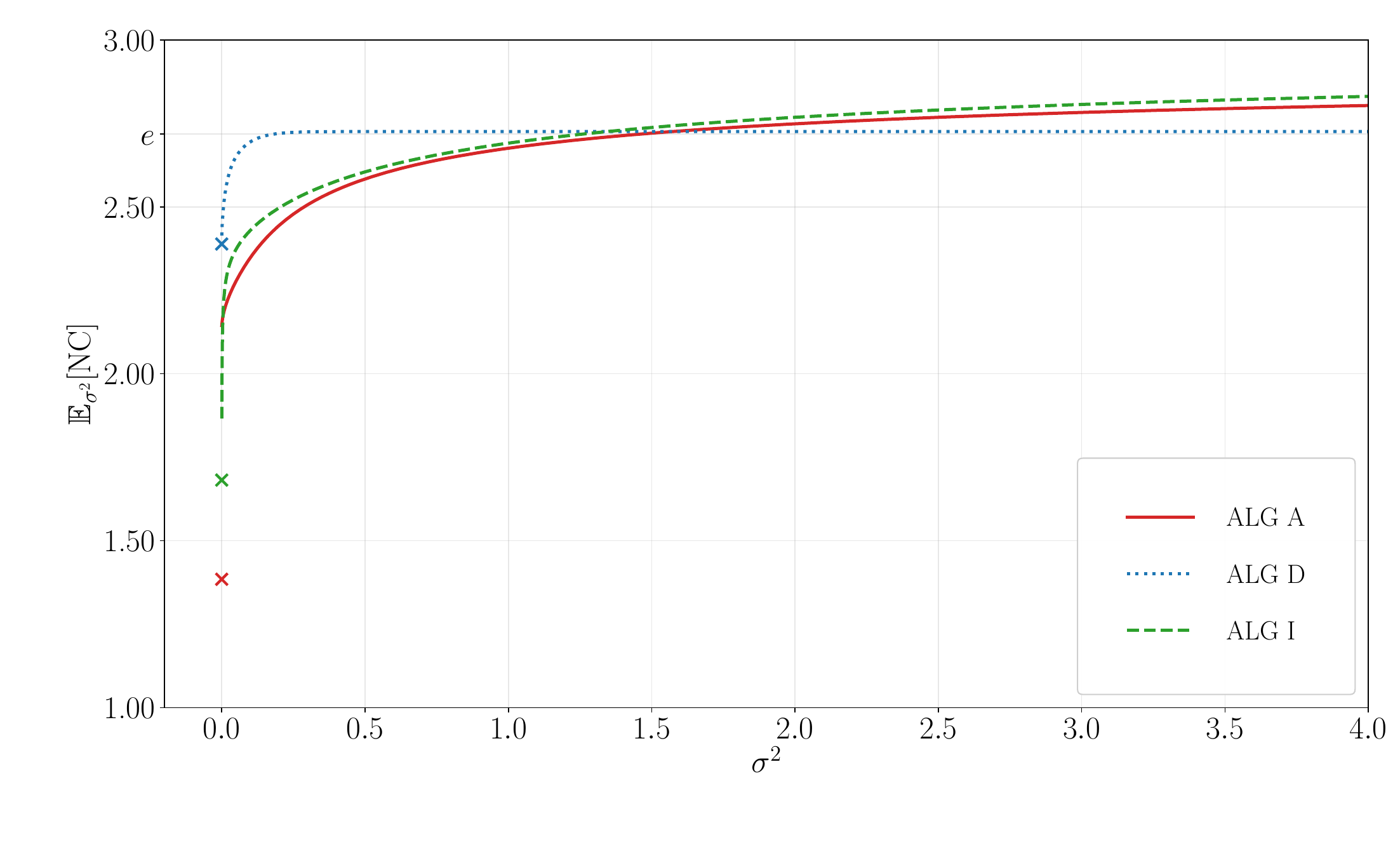}

    \smallskip
    \textit{$R=3$}

    \medskip

    \includegraphics[width=0.8\linewidth]{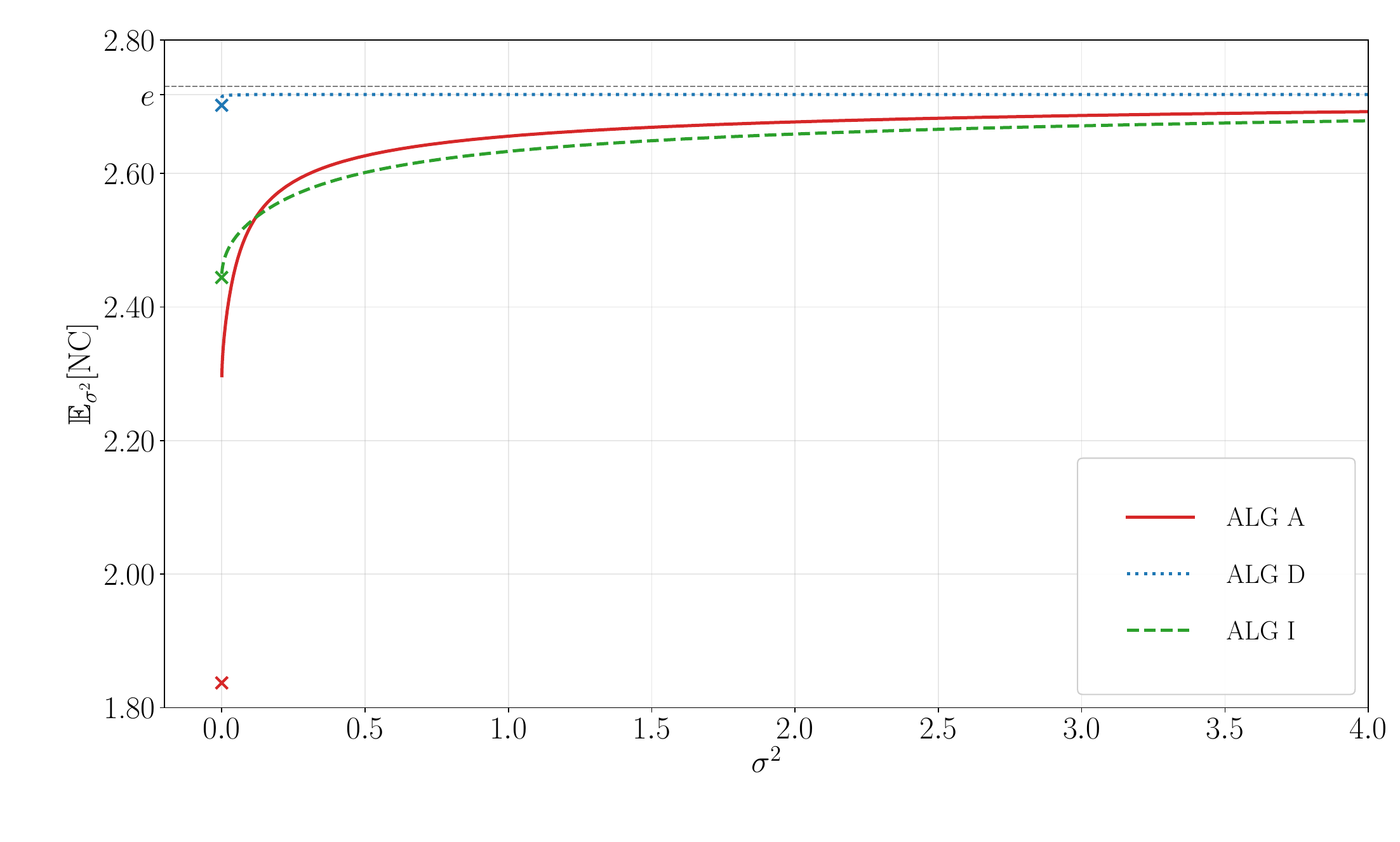}

    \smallskip
    \textit{$R=2.73$}

    \caption{Expected normalized cost as function of $\sigma^2$ for $R=3$ and $R=2.73$.}
    \label{fig:app_smoothness}
\end{figure}

\cref{fig:app_smoothness} depicts the expected normalized cost of the different bidding algorithms for $R=3$ and $R=2.73$, as a function of the variance $\sigma^2$, following the same experimental setting as in the main paper. 
%The same qualitative picture as in Section~\ref{sec:numerical} emerges for both values of $R$, with only a mild shift in the range of noise levels for which our algorithm performs best. 
As $R$ gets closer to $e$, algorithm
{\sc A} improves over {\sc D} on a wider interval of error, namely for $\sigma^2 \lesssim 40$. On the other hand, {\sc A} improves over {\sc I} only when the prediction oracle is fairly accurate.

\section{Omitted material from Section~\ref{sec:incremental}}
\label{app:medians}

Our general approach is the following.
First, we compute an approximation $\widetilde{\text{opt}}_k$ of $\text{opt}_k$, using the Python class \texttt{KMedoids} from the package \texttt{sklearn\_extra.cluster}\footnote{Released under the 3-Clause BSD license}, and their corresponding sets $\widetilde{F}_k$. These values $\widetilde{F}_n ,\ldots ,\widetilde{F}_1$ form a universe $\mathcal{U}$ of increasing values. A bidding sequence is projected on $\mathcal{U}$, in the following sense. It translates into a set of indices $\mathcal{K} \subseteq\left\{1 ,\ldots , n\right\}$ such that for every $i \in \mathcal{K}$, there is a bid $x$ in this sequence with $\widetilde{F}_i \leq x < \widetilde{F}_{i - 1}$. The algorithm constructs sets $F_k$ for every $k \in \mathcal{K}$. First, it initializes $F_n =\widetilde{F}_n =\mathcal{U}$. Then for every $k \in \mathcal{K} \setminus \left\{n\right\}$ in decreasing order, $F_k$ is obtained from the previously built set $F_{\ell}$, by projecting each point $p \in \widetilde{F}_k$ to the closest point in $F_{\ell}$, hence approximating the near-optimal set $\widetilde{F}_k$ with the only available points due to the \textit{incremental} constraint of the problem. Intermediate sets $F_i$ for $k < i < \ell$ are simply set to $F_k$. Here, we adhere to the original description of the algorithm of~\cite{chrobak2008incremental}, where an incremental solution $F_1\subseteq F_2 \subseteq \ldots \subseteq F_n$ only needs to satisfy $|F_i|\leq i$. Another option would have been to fill the sets greedily with points so to satisfy $|F_i| = i$ for all $i$.

The theoretical performance guarantee of the bidding sequence translates to a similar guarantee for the incremental medians problem, in the following sense. Suppose that the approximation ratio $\widetilde{\text{opt}}_k / \text{opt}_k$ is upper bounded by some factor $c > 1$. Unfortunately the documentation does not specify it~\cite{kmedoids}, but since we use the \texttt{kmeans++} initialization option, we can assume $c = O(\log n)$, see \cite{Arthur-Vassilvitskii-kmeansplusplus-2007}. If the bidding sequence guarantees a normalized cost $\leq\beta$, then the produced incremental medians solution guarantees a normalized cost $2 \beta \cdot c$, where $c$ comes from the approximation of $\widetilde{\text{opt}}$, and the factor $2$ comes from the triangular inequality of the distance function and the use of the projection in the algorithm.

\begin{lemma}
 Let $\text{B} =(X_i)_{i \in \ZZ}$ be a (randomized) sequence for the online bidding problem, whose robustness is at most $\beta$ and whose consistency is at most $C$. Then, for every $k$, there exists a projection onto the cost set $U =\{\cost(F_k^{*}), k\leq n\}$ of $\text{B}$ such that there exists an incremental median algorithm that is $2 \beta$-robust and achieves a normalized cost of at most $2 C$ at point $k$.
\end{lemma}

\begin{proof}We denote by $t \in \RR$ the value such that $\NC(\text{B}, t) \leq C$. Let $k \leq n$, and define $\widetilde{\text{B}} = (X_i \cdot \frac{\cost(F_k^{*})}{t})_{i \in \ZZ}$. By scale invariance, this sequence remains $\beta$-competitive, and moreover $\NC(\widetilde{\text{B}}, \cost(F_k^{*})) \leq C$. It then suffices to consider the projection of $\widetilde{\text{B}}$ onto the set $U$, as we defined. The conclusion then follows from~\cite[Lemma 5]{chrobak2008incremental}.
\end{proof}

Figure~\ref{fig:app_medians} depicts the empirical average approximation ratios of the various algorithms for predicted numbers of facilities $\hat{k}=20$ and $\hat{k}=4000$, respectively. The plots exhibit the same relative performance among the algorithms as observed for $\hat{k}=2500$, in the main paper.

\begin{figure}[t]
    \centering
    \includegraphics[width=0.8\linewidth]{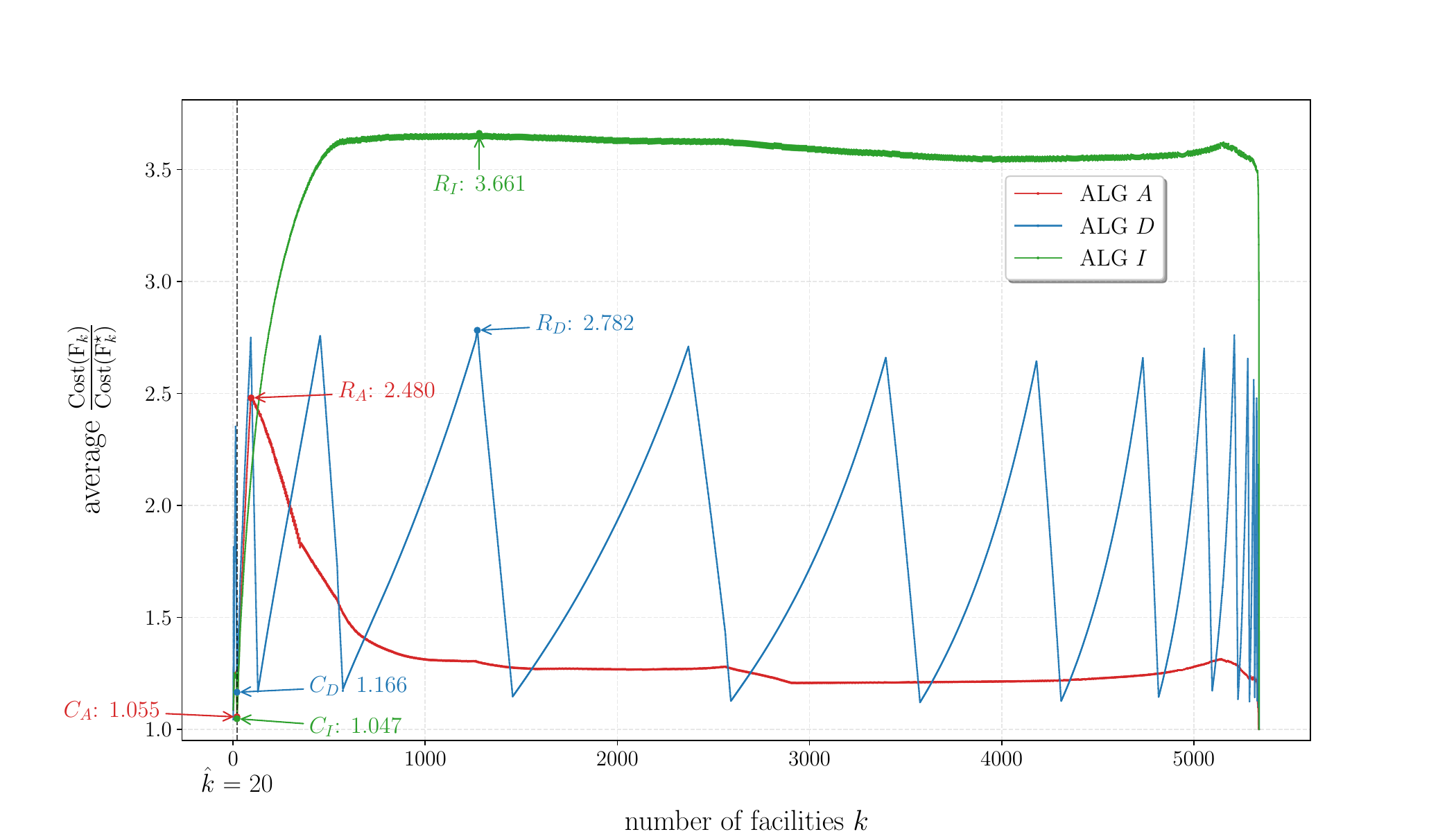}

    \smallskip
    \textit{Average empirical approximation ratios for $\hat{k}=20$.}

    \medskip

    \includegraphics[width=0.8\linewidth]{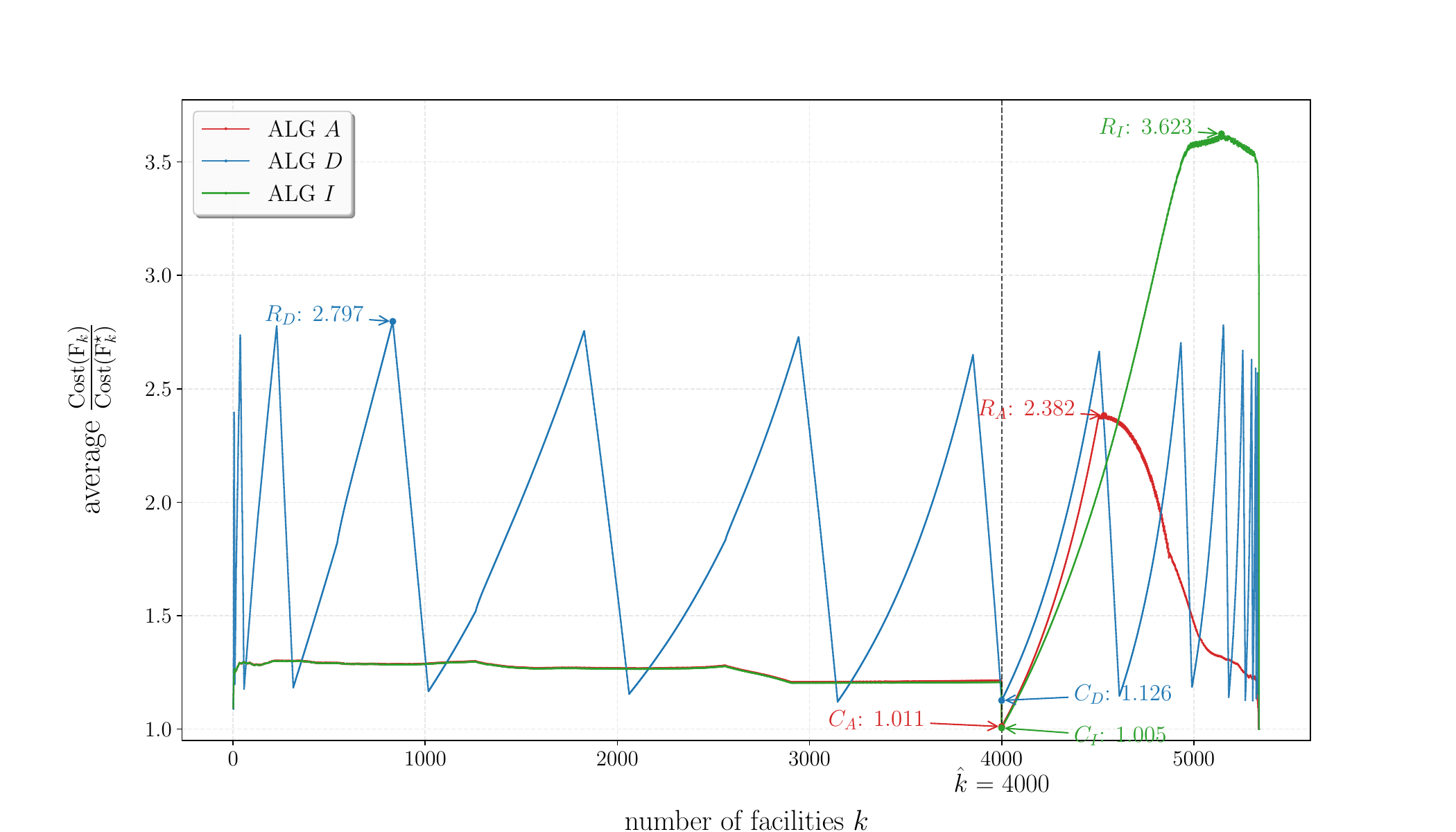}

    \smallskip
    \textit{Average empirical approximation ratios  for $\hat{k}=4000$.}

    \caption{Average empirical approximation ratios of the incremental median algorithms  for different values of $\hat{k}$.}
    \label{fig:app_medians}
\end{figure}